\documentclass[a4paper,12pt]{article}

\usepackage{relsize}
\usepackage{epigraph}
\usepackage{nicefrac}
\usepackage[margin=1in]{geometry}
\usepackage{fullpage}
\usepackage[T1]{fontenc}

\usepackage{hyperref}
\usepackage{color}
\hypersetup{colorlinks=true,linkcolor=blue,citecolor=blue}
\usepackage{float}
\usepackage{varwidth}

\newcommand{\q}{$``?"$}

\usepackage{mathtools}
\usepackage{float}
\usepackage[titletoc,title]{appendix}
\usepackage[classfont=sanserif,langfont=roman,funcfont=italic]{complexity}
\usepackage{caption}
\usepackage{comment}
\usepackage{enumerate}
\usepackage{amsmath, amsthm, amstext,  graphicx,amsopn} 
\usepackage{amssymb}
\usepackage{subfigure}
\usepackage{tcolorbox}

\newtheorem{proposition}{Proposition}

\newtheorem{corollary}{Corollary}

\newtheorem{todo}{TODO}

\newtheorem{theorem}{Theorem}
\newtheorem{lemma}{Lemma}

\newtheorem{remark}{Remark}
\usepackage{enumitem}
\usepackage{bbm}

\DeclareMathOperator{\Ex}{\mathbb{E}}

\newcommand\bigpar[1]{\bigl(#1\bigr)}

\newcommand\bigabs[1]{\bigl|#1\bigr|}

\title{Consensus with Bounded Space and Minimal Communication}
\author{Simina Br\^anzei\thanks{Purdue University, USA. E-mail: \url{simina.branzei@gmail.com}.} \and Yuval Peres\thanks{E-mail: \url{yuval@yuvalperes.com}.}}

\date{}

\begin{document}
\maketitle

\begin{abstract}
	Population protocols are a fundamental model in distributed computing, where many nodes with bounded memory and computational power have random pairwise interactions over time. This model has been studied in a rich body of literature aiming to understand the tradeoffs between the memory and time needed to perform computational tasks.
	
	We study the population protocol model focusing on the communication complexity needed to achieve consensus with high probability. When the number of memory states is $s = O(\log \log{n})$, the best upper bound known was given by a protocol with $O(n \log{n})$ communication, while the best lower bound was $\Omega(n \log(n)/s)$ communication. 
	
	We design a protocol that shows the lower bound is sharp, solving an open problem from \cite{comm_init}. When each agent has  $s=O(\log{n}^{\theta})$ states of memory, with $\theta \in (0,1/2)$, consensus can be reached in time $O(\log(n))$ with $O(n \log{(n)}/s)$ communications with high probability.
\end{abstract}

\newpage 
\section{Introduction}
Population protocols are a basic model in distributed computing introduced in~\cite{AADFP04}, where a collection of agents with bounded memory have random pairwise interactions over time. Population protocols can be used to model colonies of insects such as ants and bees, flocks of birds~\cite{AADFP04,DFGR06}, chemical reaction networks~\cite{CardelliHDC17}, gene regulatory networks~\cite{BB04}, wireless sensor networks~\cite{DV12}, and opinion formation in social networks~\cite{PVV09} (see overview in~\cite{AAEGR17}). 

In a population protocol there are $n$ agents, each of which is represented as a finite state machine and has a Poisson clock with unit rate. When the clock of an agent rings, the node wakes up and gets matched with a random other node. The two nodes update their states as a function of both of their previous states. 
In the consensus (majority) problem, each agent starts with an initial belief bit and the goal is to have all the agents learn, through interactions, the belief with higher initial count. Note this model is equivalent to a discrete model where in each round, a random node wakes up and is matched to another random node. \footnote{The equivalence between the continuous and discrete models holds up to rescaling the time.}

The population protocol model can be used to provide insights into how eusocial insects solve problems such as searching for a new home (reaching consensus), foraging for food, or allocating tasks to workers~\cite{Radeva17}.
In the case of chemical reaction networks, assumptions about chemical solutions are very similar to the rules that define population protocols: molecules can be seen as agents that have pairwise interactions (collisions), the next pair to interact is chosen randomly (as in a well-mixed solution), each molecule has a finite number of states, and each interaction can update the state of one or both molecules (see~\cite{CardelliHDC17}).
More recent experimental evidence suggests that population protocols can be implemented at the molecular level by DNA nucleotides and are equivalent to computations carried out by living cells to function correctly (see discussion in ~\cite{AAEGR17} and ~\cite{CDSPCSS13,CC12}).

The consensus problem has been studied extensively from a computational point of view, aiming to understand the resources required to reach consensus, such as time and memory~\cite{AngluinAER07,DraiefV12,AAEGR17,AlistarhDKSU17,AlistarhAG18,AlistarhG18,BerenbrinkGK20}. One of the simplest consensus protocols is the \emph{three state protocol}~\cite{AAE08}, where every node has only $3$ states of memory, labeled ``$0$”, ``$1$”, and ``$?$”. Each node starts by initializing its state to their initial belief bit. When the clock of a node $i$ rings, the node wakes up and gets matched with another random node $j$, with the following update rules. If both nodes have the same belief in $\{0,1\}$, then they keep their beliefs. However, if node $i$’s belief is $b \in \{0,1\}$ and node $j$’s belief is ``$?$'', then node $j$ updates its belief to $b$. If on the other hand nodes $i$ and $j$ have opposing beliefs in $\{0,1\}$, then node $j$ changes their belief to ``$?$'' while node $i$ keeps its belief. The three state protocol converges to the correct majority bit with high probability in time $O(\log n)$ and communication $O(n \log{n})$. \cite{5062181} gives an elegant approach for studying the three state protocol, by considering the deterministic system obtained when the number of nodes $n$ goes to infinity and studying the resulting differential equations. 

In many distributed settings, such as blockchain, setting up a communication channel between the nodes is expensive, so light communication is crucial for the efficiency of a protocol~\cite{comm_init}. Minimizing the communication complexity (or cost) requires that nodes do not communicate each time they wake up. Rather, when a node wakes up it decides first whether to communicate or not. If it decides to communicate, then it gets matched with a random other node as usual and they exchange states. If it decides to not communicate, then the node can update its internal state by itself. For example, a useful type of update that a node can do individually is to increase a counter that tells it how many rounds it has been since it last communicated. Thus the goal is to minimize the number of communications required to reach consensus with high probability. \footnote{Thus each time two nodes communicate will count as one communication, regardless of the amount of information passed through the channel.}

The study of communication in population protocols was initiated in~\cite{comm_init}, which showed that when the number $s$ of memory states of a node is $s > \log \log{n}$, then consensus can be reached with $O(n)$ communication. When $s < \log \log{n}$, the upper and lower bounds did not match; the best upper bound known is given by an algorithm with $O(n \log{n})$ communication, while the lower bound is $\Omega(n \log(n)/s)$ communication. 

Our main result solves the open problem from ~\cite{comm_init} by showing that the lower bound is tight. We design and analyze a protocol that achieves consensus with $O(n\log(n)/s)$ communication in $O(\log{n})$ time w.h.p., where $s = O(\log \log(n)^3)$ is the number of memory states per agent. Note the three state protocol described above also runs in time $O(\log{n})$ and in fact this is the minimum possible time for \emph{any} consensus protocol that is correct.

\begin{theorem}[Main Theorem] \label{thm:main}
	Suppose that each node has $s=O(\log{n}^{\theta})$ states of memory where $\theta \in (0,1/2)$. Then with high probability, consensus can be reached in time $O(\log(n))$ with $O(n \log{(n)}/s)$ communications.
\end{theorem}

\begin{remark}
When $s < \log{\log{n}}$, the upper bound from Theorem~\ref{thm:main} matches up to constants the lower bound from \cite{comm_init}. For $s < \left( \log{\log{n}}\right)^3$, Theorem~\ref{thm:main} improves the state of the art, while for $s >  \left( \log{\log{n}}\right)^3$ there is an algorithm with $O(n)$ communication. 
\end{remark}

\paragraph{Algorithm} 
Our first main contribution is to design the following consensus algorithm, which achieves the communication and time bounds of Theorem~\ref{thm:main}:


\begin{itemize}
	\item Self-select a number of $n/s$ \emph{leaders}\footnote{This step can be implemented in a decentralized way by using the randomness in the interactions. See ~\cite{comm_init} for how this is done. Note this has a negligible effect on the protocol and proofs work the same.}. The rest of the nodes will be \emph{followers}.
	\item Each follower keeps a counter that can take a value from $\{1, \ldots, 8 s+1\}$. The counter value of each follower is initialized uniformly at random from $\{1, \ldots, 8s\}$. We denote by bin $j$ the set of followers with counter value $j$.
	The followers are divided in two groups:
	\begin{itemize}
		\item \emph{Informed nodes:} with counter value in $\{1, \ldots, 8 s \}$. When the clock of such a node rings, the node increases its counter value by one.
		\item \emph{Uninformed nodes:} with counter value $8 s+1$. When the clock of such a node rings, they communicate, and if they reach an informed node, then they adopt its bit and counter value.
	\end{itemize}
	\item Each leader has a belief, which can be $0$, $1$, or \q. When the clock of a leader rings, the leader acts as follows depending on its belief:
	\begin{itemize}
		\item 0 or 1: If the leader gets matched with an informed follower, then the leader flips a fair coin. If the coin turns heads, then the leader pushes, and if it turns tails, then the leader pulls from the follower.
		
		In case of a pull, if the bits match, then the leader keeps its belief, and if they don't match, then the leader switches to \q.
		
		\item \q: If the leader meets an informed node, then it adopts the informed node's bit.
	\end{itemize}
\end{itemize}

To understand this protocol, we study the following deterministic approximation, which can be understood as taking the number of nodes $n$ to $\infty$ while keeping the parameter $s$ fixed. The deterministic system will be given by some differential equations that we analyze and then show the random system will closely approximate it for large enough $n$. We keep track of the following fractions, the denominator for all of which is the number of nodes $n$ (and then take the limit of $n \to \infty$):
\begin{description}
	\item[$\bullet$] $\alpha(t)$: Fraction of nodes that are leaders with the incorrect bit at time $t$.
	\item[$\bullet$] $\delta(t)$: Fraction of nodes that are leaders and have belief \q.
	\item[$\bullet$] $\beta_j(t)$: Fraction of nodes that are followers, informed, and have the incorrect bit and counter value $j$.
	\item[$\bullet$] $\beta(t) = \sum_{j=1}^{8s} \beta_j(t)$: Fraction of nodes that are followers, informed, and have the incorrect bit.
	\item[$\bullet$] $\gamma_j(t)$: Fraction of nodes that are followers with counter value $j$.
	\item[$\bullet$] $u(t)$: Fraction of nodes that are uninformed.
\end{description}

We let $\widetilde{\alpha}(j)$ be the number of leaders with the incorrect bit after $j$ clock rings in the random system given by our protocol.
The execution of our protocol for $n=6000$ nodes with $s=5$ is illustrated in Figure 1. 

\begin{figure}[h!]
	\centering
	\subfigure[$\frac{\widetilde{\alpha}(t)}{n} \cdot s$ (in red) and $\frac{\widetilde{\beta}(t)}{n} \cdot \frac{s}{s-1}$  (in blue) over time.]
	{
		\includegraphics[scale=0.56]{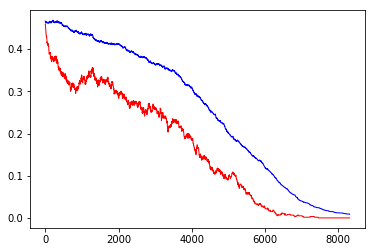}
	}
	\subfigure[${{\alpha}(t)} \cdot s$ (in red) and ${{\beta}(t)} \cdot \frac{s}{s-1}$ (in blue) over time.]
	{
		\includegraphics[scale=0.56]{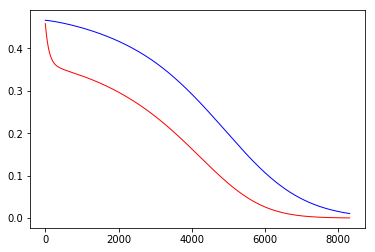}
	}
	\subfigure[$\frac{\widetilde{\delta}(t)}{n} \cdot s$ over time.]
	{
		\includegraphics[scale=1.5]{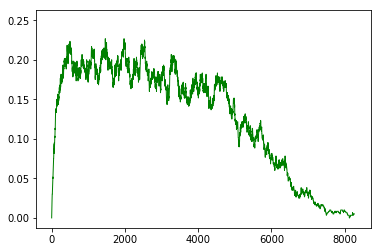}
	}
	\subfigure[${{\delta}(t)} \cdot s$ over time.]
	{
		\includegraphics[scale=1.5]{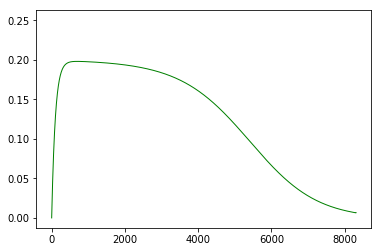}
	}
	\caption{The random and deterministic system for $n=3000$, $s=5$, and initial minority $45\%$ (that is, $\rho=0.1$).}
\end{figure}

\subsection{Overview}

To prove Theorem 1, it suffices to show that w.h.p. consensus is reached in time $O(\log(n)$ and the rate of communication is $O(1/s)$. i.e. only one in $\Omega(s)$ clock rings leads to a communication. The latter claim will follow if we show that there at most $O(n/s)$ uninformed nodes throughout, since the only nodes that initiate communications are the $n/s$ leaders and the uninformed nodes.

\medskip

The proof of Theorem 1 consists of three phases:
\begin{description} 
\item[I.]	From near tie to a dominant majority in time $\Theta(s)$;
\item[II.]	Exponential decay of minority at rate $\Omega(1/s)$ for $\Theta(s^2)$ time units;
\item[III.]	Exponential decay of minority at constant rate for $\Theta(\log n)$ time units.
\end{description}
Each phase is first proved for a deterministic approximation to the process that solves a system of $\Theta(s)$ coupled differential equations. Then martingale arguments, described further below, enable us to analyze the random system.
The system of differential equations is too involved to solve directly, so four different potential functions are used to control it.

\medskip

{\bf Phase I} is the most delicate. The key to this phase is to change variables and consider the relative advantage of the majority among decisive leaders and informed followers. Thus the new variables are 
$$\xi=\frac{\left(1/s-\delta-\alpha\right)-\alpha}{1/s-\delta} \quad \text{and} \quad 
\eta_j=\frac{(\gamma_j-\beta_j)-\beta_j}{\gamma_j} \,.
$$
It is quite intuitive (and easy to verify, see Proposition 1) that the minimal relative advantage
$$ \Phi = \min\Bigl\{  \xi, \eta_1, \ldots, \eta_{8s} \Bigr\}$$
is nondecreasing, but to bound the convergence time this must be made more quantitative. 

In our protocol, followers adopt the belief bits sent to them, but leaders only flip their belief bits after pulling contrary bits from two or more followers. This the relative advantage of the majority among leaders grows faster than the corresponding advantage among followers. Our proof uses this to define a modified potential function $\Psi$, see (\ref {psi_fun}), that grows at rate $\Omega(s/rho)$ as long as $\Phi$ is bounded away from 1, where $\rho$ is the initial value of $\Phi$. We deduce in Proposition \ref{prop:phase1_progress} that $\Phi$ will exceed a prescribed level $\lambda<1$ by time $T_1=O\left(\frac{s}{\rho(1-\lambda)}\right)$. 
	
	\medskip 
	
	{\bf Phase II}  starts at time $T_1$. In this phase (using a bound on $\beta$ obtained in phase I) we show that a linear potential function in the original variables 
	$$\Lambda_2=\alpha+\frac{\delta}{16}+\frac{\beta}{4}$$
	exhibits exponential decay at rate $\Omega(1/s)$.
	
	\medskip

{\bf Phase III} starts at time $T_2=\Theta(s^2)$. At this time, $\beta$ is exponentially small in $s$ (via phase II). This allows us to use a modified
potential function $\Lambda_3$ (a linear combination of $\alpha, \delta$ and the $\beta_j$)  that exhibits true exponential decay. We deduce that $\delta$ and all the minority fractions $\alpha,\beta_j$ will be well below $1/n$ by time $O(\log n)$. 

To handle the random fluctuations in phases I and II we use the differential equation method for approximating stochastic processes developed by Kurtz~\cite{kurtz} and Wormald~\cite{wormald}, in the optimized version presented by Warnke~\cite{warnke}. Such approximations cannot be used in phase III, because the approximation error between the random process and the deterministic one scaled by $n$ is of order $\Omega(n^{1/2})$.

In the main text, we focus on the case where the initial advantage $\rho=\Phi(0)$ is a positive constant. However, even if $\rho$ tends to zero with $n$, the protocol works and reaches consensus on the majority value w.h.p provided 
$\rho>n^{-\sigma}$ for some $\sigma \in (0,1/2)$. See the discussion in the concluding remarks.



\subsection{Related literature} \label{sec:further_related_work}

\cite{AngluinAER07} studies the computational power of population protocols and gives a precise characterizations of the class of predicates that can be computed in a stable way on complete interaction graphs (i.e. where each agent can get matched with any other agent).

The tradeoffs between time and memory required to reach consensus have been studied in a rich body of literature. \cite{AAEGR17} studies majority and leader election, showing a  unified lower bound for the two problems that relates the space available per node with the time complexity achievable by a protocol. In particular, solving these tasks when each agents has $O(\log \log{n})$ memory states takes $\Omega(n/\polylog(n))$ time in expectation. On the other hand both tasks can be solved in polylogarithmic time with $O(\log(n)^2)$ memory states per agent. \cite{AlistarhAG18} shows a lower bound of $\Omega(\log{n})$ states for any consensus protocol that stabilizes in time $O(n^{1-c})$, for any constant $c > 0$. This result is complemented by a consensus protocol that uses $O(\log{n})$ states and stabilizes in $O(\log^{2}{n})$ time.

\cite{PVV09} studied the three state protocol and suggested a symmetric variant that converges faster, where both the initiator and responder update their states. 
Their paper shows the differential equations that characterize the three state protocol and inspired the approach in the present paper.
\cite{DraiefV12} studies the problem of exact consensus, where the protocol must converge to the correct final state with probability $1$ and the goal is bound the expected convergence time. Their work uses four memory states per agent to design a protocol that achieves consensus with probability $1$ from any starting configuration. 

Plurality consensus is the problem where initially each node has a color chosen from a set $\{1, \ldots, k\}$ and one of the colors has an initial advantage. The goal is that all the nodes eventually adopt the color found in highest proportion initially; bounds on the convergence time were given in~\cite{BecchettiCNPS15,BecchettiCNPST17}.

In the leader election problem, all the agents start from the same state, and the goal is that they reach an outcome where exactly one agent is in a special leader state. The time and memory required for leader election in population protocols were studied in~\cite{DotyS18}, which showed a lower bound of $\Omega(n)$ on the time required to select a leader with probability $1$, when each agent has a constant amount of space. \cite{AlistarhG15} showed that  poly-logarithmic stabilization time can be achieved by allowing $O(\log^3{n})$ states rather than constant number of states per agent. \cite{SOIKM20} designed a protocol that runs in $O(\log{n})$ parallel time in expectation with $O(\log{n})$ states per agent. 
\cite{BerenbrinkGK20} designs and analyzes a population protocol that uses $\Theta(\log \log{n})$ states per agent, and elects a leader in $O(n \log{n})$ interactions in expectation. 

\cite{MSA19} studies the proportion computation problem, where each agent starts in one of two states $A$ and $B$, and the goal is that each agent learns the fraction of agents that initially started in state $A$.
\cite{CGKK+15} studies Lotka-Volterra type of dynamics in the population protocol model,
where agents have types and the update rule is such that when an agent $a$ of type  $i$ interacts with an agent $b$ of type $j$ with $ a$ as the initiator, then $b$’s type becomes $i$ with probability $P_{ij}$. This update rule is a simple variant of the Lotka-Volterra update rule and can be used to model opinion dynamics in social networks. The work in \cite{CGKK+15} shows that any such protocol converges in time polynomial in $n$ when the interaction graph is complete, while convergence time can be exponential when interaction is restricted (e.g. the graph is a star).

\section{Deterministic system} \label{sec:deterministic_system_def}

We start by studying the deterministic dynamical system obtained by taking $n \to \infty$ in the algorithm. Recall the variables defined earlier for the deterministic system. We introduce a few auxiliary variables and deduce the differential equations for the deterministic system.

\medskip 

\noindent \textbf{Notation}: Let $\Gamma(t)$ denote the fraction of informed followers at time $t$ $$\Gamma(t) = 1 - \frac{1}{s}- u(t) = \sum_{j=1}^{8s} \gamma_{j}(t)$$ and $R(t)$ the rate at which informed nodes depart each bin $j$ (see equation \ref{eq:def_beta_j_dot}) $$R(t) = 1 + \frac{1}{2s} - \frac{\delta(t)}{2} - u(t).$$

\begin{lemma} \label{lem:def_deterministic}
The deterministic system obtained by taking $n \to \infty$ in the algorithm can be described by the following differential equations:
\begin{align}
\dot{\alpha} & = - \frac{\alpha}{2} \cdot \left(\Gamma - \beta \right) + \delta \cdot \beta \label{alpha_dot} \\
\dot{\delta} & = \frac{\alpha}{2} \cdot \left(\Gamma - \beta \right) + \frac{1/s - \alpha - \delta}{2} \cdot \beta - \delta \cdot \Gamma \\
\dot{u} & = \gamma_{c \cdot s} - u \cdot \Gamma \label{u_dot} \\
\label{eq:def_beta_j_dot}
\dot{\beta}_j & = \beta_{j-1} - \beta_j \cdot R, \; \; \forall j \in\{2, \ldots, 8s\}\\
\label{eq:def_beta_1_dot}
\dot{\beta}_1 & = -\beta_1 \cdot R + \frac{\alpha}{2} \cdot \Gamma \\
\dot{\gamma}_j & = \gamma_{j-1} - \gamma_j \cdot R, \; \; \forall j \in \{2, \ldots, 8s\} \label{gamma_j_dot} \\
\dot{\gamma}_1 & = - \gamma_1 \cdot R + \left(\frac{1}{2s} - \frac{\delta}{2}\right) \cdot \Gamma \label{gamma_1_dot}
\end{align}
\end{lemma}
\begin{proof}
	Note that the fraction of correct leaders is $1/s - \alpha(t) - \delta(t)$ and the fraction of correctly informed followers is $1 - 1/s - u(t) - \beta(t)$.
	
	The identity for $\dot{\alpha}$ holds since the first term in the update for $\dot{\alpha}(t) $ counts the expected number of leaders with the incorrect bit at time $t$ that pull from an informed node with the correct bit. The second term counts leaders with \q that pull from an informed node with the wrong bit.
	
	The identity for $\dot{\delta}$ holds since the first term counts wrong leaders that pulled from a correct informed node, the second term counts correct leaders that pulled from an incorrect informed node, and the last term counts leaders with \q that pulled from an informed node.
	
	
	For $\dot{\beta}_1$ we get the following identities, where the first term counts the followers with counter value $1$ that get a clock tick and increase their counter to 2, the second term counts the followers that get pushed a correct bit from the leaders (and so leave the set of followers with the wrong bit and counter value 1), the third term counts the followers that get pushed the wrong bit from leaders with incorrect information, and the last term counts uninformed followers that pull the wrong bit from followers with counter value 1. Thus
	\begin{align}
	\dot{\beta}_1 & = - \beta_1 - \frac{1/s - \delta - \alpha}{2} \cdot \beta_1 + \frac{\alpha}{2} \cdot \left( \Gamma - \beta_1 \right) + u \cdot \beta_1 \\
	& = -\beta_1 \cdot \Bigl( 1 + \frac{1}{2s} - \frac{\delta}{2} - u\Bigr) +  \frac{\alpha}{2} \cdot \Gamma \\
	& = -\beta_1 \cdot R + \frac{\alpha}{2} \cdot \Gamma
	\end{align}
	
	For $\dot{\beta}_j$, where $j \in \{2, \ldots, 8s\}$, we get the update rule
	\begin{align}
	\dot{\beta}_j & = \beta_{j-1} - \frac{1/s - \delta}{2} \cdot \beta_j + u \beta_j, \; \; \forall j \in\{2, \ldots, 8s\} \notag \\
	& = \beta_{j-1} - \beta_j \cdot R
	\end{align}
	
	For $\dot{\gamma}_j$, where $j \in \{2, \ldots, 8s\}$, we get the update rule
	\begin{align}
	\dot{\gamma}_j & = \gamma_{j-1} - \gamma_j \cdot \Bigl( 1 + \frac{1}{2s} -  \frac{\delta}{2} - u\Bigr) \notag \\
	& = \gamma_{j-1} - \gamma_j \cdot R
	\end{align}
	
	For $\dot{\gamma}_1$, we get the update rule
	\begin{align}
	\dot{\gamma}_1 & =  - \gamma_1 + \left(\frac{1}{2s} - \frac{\delta}{2} \right) \cdot \left( \Gamma - \gamma_1 \right) + u \cdot \gamma_1 \notag \\
	 &  = - \gamma_1 \cdot R + \left(\frac{1}{2s} - \frac{\delta}{2}\right) \cdot \Gamma
	\end{align}
\end{proof}

\begin{figure}[h!]
	\centering
	\subfigure[$\alpha$ (in red) and $\beta$ (in blue) over time. The values are normalized, so the plot shows in fact $\alpha(t) \cdot s$ and $\beta(t) \cdot s/(s-1)$.]
	{
		\includegraphics[scale=0.56]{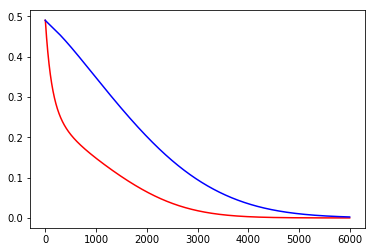}
		\label{fig:alpha_beta_deterministic}
	}
	\subfigure[$\delta$ over time]
	{
		\includegraphics[scale = 0.56]{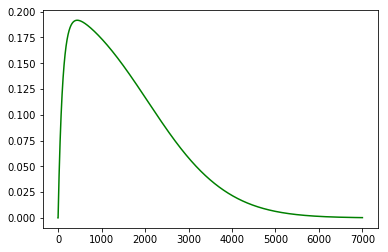}
		\label{fig:delta_deterministic}
	}
	\subfigure[$\gamma_j$ over time]
	{
		\includegraphics[scale=0.56]{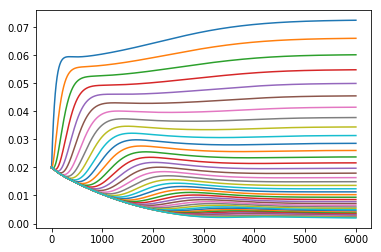}
		\label{fig:gamma_deterministic}
	}
	\subfigure[$\beta_j(t)$ over time]
	{
		\includegraphics[scale=0.56]{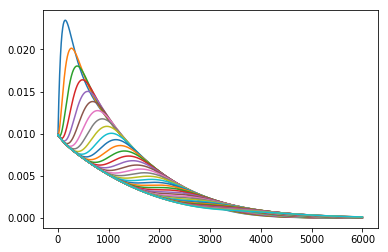}
		\label{fig:beta_j_deterministic}
	}
\subfigure[$\log \left(\frac{\alpha}{1/s - \delta} \right)$ and $\log\left(\frac{\beta_j(t)}{\gamma_j(t)}\right)$, for $j = \{1, \ldots, 8s\}$ over time. Note that $\log \left(\frac{\alpha}{1/s - \delta} \right)$ - shown in blue - is ahead of all the other variables in the plot.]
{
	\includegraphics[scale=0.7]{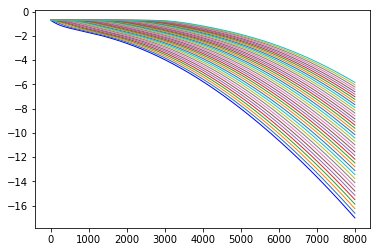}
	\label{fig:log_ratio_beta_gamma_j_deterministic}
}
	\caption{The deterministic system for $s=5$ and initial minority $49\%$ (that is, $\rho=0.02$).}
	\label{fig:alpha_beta_gamma_delta_deterministic}
\end{figure}

\begin{figure}[h!]
	\centering
	\subfigure[The advantage variables $\xi$ (in red) and $y$ (in blue) over time.]
	{
		\includegraphics[scale=0.56]{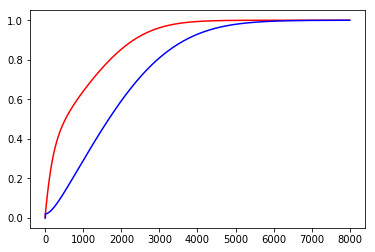}
		\label{fig:xi_y_deterministic}
	}
	\subfigure[The advantage variables $\eta_j$ over time, for $j \in \{2, \ldots, 8s\}$.]
	{
		\includegraphics[scale = 0.56]{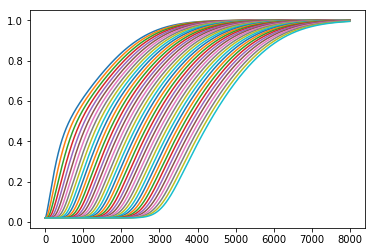}
		\label{fig:eta_j_deterministic}
	}
	\caption{The deterministic system for $s=5$ and initial minority $49\%$ (that is, $\rho=0.02$).}
	\label{fig:xi_eta_y_deterministic}
\end{figure}

We show the variables that define the deterministic system are bounded from below. For $\alpha$, we will study the variable $w := 1/s - \delta - \alpha$, while for $\beta_j$ we study $w_j := \gamma_j - \beta_j$. Then $\dot{w}_j = w_{j-1} - w_j R$.

\begin{lemma} \label{lem:bounded_differential}
Suppose $0 < \alpha(0) < 1/s - \delta(0)$ and $0 < \beta_j(0) < \gamma_j(0)$ for all $j$.
The following functions have a strictly positive derivative
\begin{align} \label{eq:positive_derivatives}
g_{\alpha}(t) & = \alpha(t) \cdot e^{\frac{t}{2}(1-\frac{1}{s})}, \, \, g_{\delta}(t) = \delta(t)\cdot  e^{t(1-\frac{1}{s})} , \, \,
g_{\beta_j}(t) = \beta_j(t) \cdot e^{t(1+\frac{1}{2s})} , \, \,
g_{u}(t) = u(t) \cdot e^{t(1-\frac{1}{s})} \notag \\
g_{\gamma_j}(t) & =  \gamma_j(t) \cdot e^{t(1+\frac{1}{2s})} , \, \,
g_{w}(t) = w(t) \cdot e^{\frac{t}{2} (1 - \frac{1}{s})} , \, \,
g_{w_j}(t) = w_j(t) \cdot e^{t(1+\frac{1}{2s})}
\end{align}
In particular, for all $t > 0$, all the variables $\alpha(t),\delta(t),\beta_j(t),u(t),\gamma_j(t),w(t),w_j(t)$ are strictly positive  and $\alpha(t) > e^{\frac{-t}{2}(1-\frac{1}{s})} \cdot \alpha(0)$.
\end{lemma}
\begin{proof}
We will use the first violation method.

Let $t \geq 0$ be the first time at which the derivative of a function in (\ref{eq:positive_derivatives}) is non-positive. Then at time $t$, we have that $u(t),\delta(t) \geq 0$ and all the other variables are strictly positive, i.e.
\begin{align}  \label{eq:positive_vars}
\alpha(t) > 0, \delta(t) \geq 0, \beta_j(t) > 0, u(t) \geq 0, \gamma_j(t) > 0, \Gamma(t) > 0, w(t) > 0, w_j(t)>0
\end{align}
We will check that all the functions in (\ref{eq:positive_derivatives}) have strictly positive derivatives at time $t$, which will give a contradiction.

\medskip

\noindent \emph{Case 1}: Consider $g_{\alpha}'(t)$. Then by (\ref{eq:positive_vars}),
\begin{align}
g_{\alpha}'(t) & = e^{\frac{t}{2}(1 - \frac{1}{s})} \cdot \left[ \frac{\alpha(t)}{2} \left( u(t) + \beta(t)\right) + \delta(t) \beta(t) \right] > 0\notag
\end{align}

\medskip

\noindent \emph{Case 2}: Consider $g_{\delta}'(t)$. Then by (\ref{eq:positive_vars})
\begin{align}
g_{\delta}'(t) & = e^{t(1-\frac{1}{s})} \cdot \left[\alpha(t) \beta(t) + \frac{\alpha(t)}{2} \Gamma(t) + \left(\frac{1}{s} - \delta(t)\right) \frac{\beta(t)}{2} + \delta(t)u(t) \right] > 0, \notag
\end{align}
where we used the fact that
 $1/s - \delta(t)> w(t)> 0$.

\medskip

\noindent \emph{Case 3}: Consider $g_{\beta_1}'(t)$. Then by (\ref{eq:positive_vars})
\begin{align}
g_{\beta_1}'(t) & = e^{t(1+\frac{1}{2s})} \cdot \left[ \beta_1(t) \cdot \left(u(t) + \frac{\delta(t)}{2} \right) + \frac{\alpha(t)}{2} \Gamma(t) \right] > 0 \notag
\end{align}

\medskip

\noindent \emph{Case 4}: Consider $g_{\beta_j}'(t)$, for $j \geq 2$. Then by (\ref{eq:positive_vars})
\begin{align}
g_{\beta_j}'(t) & = e^{t(1+\frac{1}{2s})} \cdot \left[ \beta_{j-1}(t) + \beta_j(t) \cdot \left(u(t) + \frac{\delta(t)}{2} \right) \right] > 0\notag
\end{align}

\medskip

\noindent \emph{Case 5}: Consider $g_{u}'(t)$. Using the identity $1 - 1/s - \Gamma(t) = u(t)$ and (\ref{eq:positive_vars}),
\begin{align}
g_{u}'(t) & = e^{t(1 - \frac{1}{s})} \cdot \left[\gamma_{cs}(t) + u(t)^2 \right] > 0 \notag
\end{align}

\medskip

\noindent \emph{Case 6}: Consider $g_{\gamma_1}'(t)$. Since $w(t) > 0$, we get that $1/s - \delta(t) > 0$, so
\begin{align}
g_{\gamma_1}'(t) & = e^{t(1+\frac{1}{2s})} \cdot \left[ \gamma_1(t) \cdot \left(u(t) + \frac{\delta(t)}{2}\right) + \frac{\Gamma(t)}{2} \left(\frac{1}{s} - \delta(t)\right) \right] > 0 \notag
\end{align}

\noindent \emph{Case 7}: Consider $g_{\gamma_j}'(t)$ for $j \geq 2$. Then by (\ref{eq:positive_vars})
\begin{align}
g_{\gamma_j}'(t) &=  e^{t(1+\frac{1}{2s})} \cdot \left[ \gamma_{j-1}(t) + \gamma_j(t) \cdot \left(u(t) + \frac{\delta(t)}{2}\right) \right] > 0 \notag
\end{align}

\medskip

\noindent \emph{Case 8}: Consider $g_{w}'(t)$. Since $w_j > 0$ for all $j$, we obtain that $\Gamma(t) - \beta(t) > 0$. Moreover, $1 - 1/s - \beta(t) \geq 1 - 1/s - u(t) - \beta(t) = \Gamma(t) - \beta(t) > 0$. Then by (\ref{eq:positive_vars})
\begin{align}
g_{w}'(t) & = e^{\frac{t}{2}(1-\frac{1}{s})} \cdot \left[ \delta(t) \cdot \left(\Gamma(t) - \beta(t)\right) + \frac{w(t)}{2}\left(1 - \frac{1}{s} - \beta(t)\right) \right] > 0 \notag
\end{align}

\medskip

\noindent \emph{Case 9}: Consider $g_{w_j}'(t)$ for $j \geq 2$. Then by (\ref{eq:positive_vars})
\begin{align}
g_{w_j}'(t) & = e^{t(1 + \frac{1}{2s})} \cdot \left[ w_j(t) \cdot \left(u(t) + \frac{\delta(t)}{2} \right) + w_{j-1}(t) \right] > 0
\end{align}
\medskip

\noindent \emph{Case 10}: Consider $g_{w_1}'(t)$. Then at time $t$ it still is the case that $w_1(t) > 0$; all the other variables are strictly positive as well. Then by (\ref{eq:positive_vars})
\begin{align}
g_{w_1}'(t) &= e^{t(1 + \frac{1}{2s})} \cdot \left[  w_1(t) \cdot \left(u(t) + \frac{\delta(t)}{2}\right) + w(t) \cdot \frac{\Gamma(t)}{2} \right] > 0
\end{align}
Thus all the functions in the lemma statement must have strictly positive derivatives.
\end{proof}

Let $T > 0 $ be an upper bound on time. Define a domain large enough to contain all the values that the variables may take:
\begin{align} \label{def:domain_0}
D_0 = & \Bigl\{(\vec{y}) \mid \vec{y} =
\left(\alpha, \{\beta_j\}, \{\gamma_j\}, \delta)\}\right); 0 < \alpha < \frac{1}{s} ; 0 < \beta_j < \gamma_j ; \sum_{j=1}^{8s} \gamma_j < 1 ; \frac{-1}{s} < \delta < \frac{1}{s} \Bigr\}
\end{align}

We use the following existence theorem, adapted from Hurewicz~\cite{hurewitz_book}.

\medskip

\noindent \textbf{Theorem} [\cite{hurewitz_book}, Theorem 11, pp. 32].
\emph{If in some bounded $a$-dimensional domain $D_0$, the function $F:D_0 \to {\mathbb R}$ satisfies a Lipschitz condition, then the solution of $\frac{dy}{dt}=F(y)$ passing through any point of $D_0$ may be uniquely extended arbitrarily close to the boundary of $D_0$.}

\medskip

\begin{corollary} \label{cor:hurewicz}
Suppose $0 < \alpha(0) < 1/s - \delta(0)$ and $0 < \beta_j(0) < \gamma_j(0)$ for all $j$. Then the system of differential equations in (\ref{alpha_dot}--\ref{gamma_1_dot}) with these initial conditions has a unique solution for all times $t \geq 0$.
\end{corollary}
\begin{proof}
Equation~(\ref{u_dot}) follows from equations (\ref{gamma_j_dot}) and (\ref{gamma_1_dot}) using the identity $u(t) = 1 - 1/s - \Gamma(t)$.
Lemma~\ref{lem:bounded_differential} and the preceding theorem ensure that for any $T > 0$, the system (\ref{alpha_dot}--\ref{gamma_1_dot}) has a unique solution in $[0,T)$.
\end{proof}

\section{Phase I} \label{sec:phase_1}

We start by studying the advantage of the correct nodes over the incorrect ones in several categories (leaders, nodes with counter value $j$, for $j=1, \ldots, 8s$), which are tracked by the following variables:
\begin{itemize}
	\item $\xi = \frac{1/s - \delta - 2\alpha}{1/s - \delta}$: the normalized advantage of the correct leaders over the incorrect leaders at time $t$, counted among the group of leaders without a question mark.
	\item $\eta_j = \frac{\gamma_j - 2 \beta_j}{\gamma_j}$, for $j = 1, \ldots, 8s$: the advantage of informed nodes with counter value $j$ with the correct bit over the nodes with the incorrect bit and counter $j$ at time $t$.
	\item $\eta = \frac{\Gamma - 2\beta}{\Gamma}$: the normalized advantage of the informed followers with the correct bit over the ones with the incorrect bit. Note that $\eta$ is a convex combination of the $\eta_j$'s, thus it can alternatively be written as:
	\begin{align} \label{convex_y}
	\eta = \sum_{j=1}^{8s} \frac{\gamma_j \cdot \eta_j}{\Gamma }
	\end{align}
\end{itemize}

We will show in Proposition \ref{prop:phase1_progress} that the minimum of these advantages is increasing with a rate large enough to guarantee sufficient progress in the first $O(s)$ steps of the protocol.

\begin{lemma}
The derivatives of ${\xi}$, ${\eta}_j$, and ${y}$ are given by:
\begin{align}
\dot{\xi} & = \frac{\delta \cdot \Gamma}{1/s - \delta} \cdot \bigl(\eta - \xi \bigr) + \frac{\Gamma}{4} \cdot \bigl(1 - \xi^2 \bigr) \cdot \eta \label{xi_dot_def} \\
\dot{\eta_j} & =  \frac{\gamma_{j-1}}{\gamma_j} \cdot  \bigl(\eta_{j-1} - \eta_j\bigr), \; \forall j = 2, \ldots, 8s \; \; \label{eta_j_dot_def} \\
\dot{\eta_1}(t) & = \frac{\Gamma}{2 \gamma_1} \cdot (1/s - \delta) \cdot \bigl(\xi - \eta_1\bigr) \label{eta_1_dot_def}\\
\end{align}
\end{lemma}
\begin{proof}
By definition, ${\xi} = 1 - \frac{2\alpha}{1/s - \delta}$. Rewriting $\alpha$, $\beta$, and $\beta_j$ in terms of the advantage variables $\xi$, $\eta$, and $\eta_j$ gives:
\begin{align}
2 \beta_j & = (1 - \eta_j) \gamma_j , \; \; \forall j \in \{1, \ldots, 8s\}  \notag \\
2 \alpha & = (1 - \xi) \left(\frac{1}{s} - \delta\right) \; \mbox{ and } \;
2 \beta = (1 - \eta) \Gamma
\end{align}
Differentiating $\xi$ gives:
\begin{align}
\left(\frac{1}{s} - \delta\right)^2 \cdot \dot{\xi} & = - 2 \dot{\alpha} \cdot \left(\frac{1}{s} - \delta\right) + \left(\frac{1}{s} - \delta\right)' \cdot 2 \alpha \notag \\
& = \left[ \alpha(\Gamma - \beta) - 2 \delta \beta \right] \cdot \left(\frac{1}{s} - \delta\right) - \alpha \cdot \left[ \alpha(\Gamma - \beta) + \beta \left(\frac{1}{s} - \delta - \alpha\right) - 2 \delta \Gamma \right] \notag \\
& = \left(\frac{1}{s} - \delta\right) \cdot \left[ \frac{1 - \xi}{2} \left( \eta \Gamma \left(\frac{1}{s} - \delta\right) \frac{1 + \xi}{2} + 2 \delta \Gamma \right) - \delta(1 - \eta) \Gamma \right]
\end{align}
Thus
\begin{align}
\left(\frac{1}{s} - \delta\right) \dot{\xi} & = (1 - \xi) \delta \Gamma - \delta(1 - \eta) \Gamma + \frac{1 - \xi^2}{4} \eta \Gamma \left(\frac{1}{s} - \delta\right) \notag \\
& = \delta \Gamma ( \eta - \xi) + \frac{1 - \xi^2}{4} \eta \Gamma  \left(\frac{1}{s} - \delta\right)
\end{align}
Dividing both sides by $\left(\frac{1}{s} - \delta\right)$ gives
\begin{align} \label{xi_dot_final}
\dot{\xi} = \frac{\delta \Gamma}{\left(\frac{1}{s} - \delta\right)} (\eta - \xi) + \frac{1 - \xi^2}{4} \cdot \Gamma \eta
\end{align}
Differentiating $\eta_j$ gives:
\begin{align} \label{eta_j_dot_def_direct}
\gamma_j^2 \cdot \dot{\eta}_j = - 2 \dot{\beta}_j \cdot \gamma_j + 2 \beta_j \cdot \dot{\gamma}_j
\end{align}
For $j \geq 2$, equation (\ref{eta_j_dot_def_direct}) can be rewritten as
\begin{align}
\gamma_j^2 \cdot \dot{\eta}_j & = -2 \gamma_j \cdot (\beta_{j-1} - \beta_j \cdot R) + 2 \beta_j \cdot (\gamma_{j-1} - \gamma_j \cdot R) \notag \\
& = - 2 \gamma_j \cdot \beta_{j-1} + 2 \beta_j \cdot \gamma_{j-1} \notag \\
& = - \gamma_j \cdot (1 - \eta_{j-1}) \cdot \gamma_{j-1} + \gamma_{j-1} \cdot (1 - \eta_j) \cdot \gamma_j \notag \\
& = \gamma_j \cdot \gamma_{j-1} \cdot \Bigl( \eta_{j-1} - \eta_j \Bigr)
\end{align}
Therefore $\dot{\eta}_j = \frac{\gamma_{j-1}}{\gamma_j} \cdot \Bigl( \eta_{j-1} - \eta_j\Bigr)$ as required.

\medskip

For $\eta_1$, the derivative satisfies the identities
\begin{align}
\gamma_1^2 \cdot  \dot{\eta}_1 & = -2 \dot{\beta}_1  \gamma_1 + 2 \beta_1  \dot{\gamma}_1 \\
& = - \gamma_1 \cdot \bigl[ -2 \beta_1 \cdot R + \alpha \cdot \Gamma \bigr] + 2 \beta_1 \cdot \Bigl[ - \gamma_1 \cdot R + \left( \frac{1}{2s} - \frac{\delta}{2} \right) \Gamma \Bigr] \\
& = - \alpha \Gamma \cdot \gamma_1 + \beta_1 \Gamma \left(\frac{1}{s} - \delta \right) \\
& = \frac{\Gamma}{2} \cdot \Bigl[ -2 \alpha \gamma_1 + 2 \beta_1 \left(\frac{1}{s} - \delta \right)  \Bigr] \\
& = \frac{\Gamma}{2} \cdot \Bigl[ - \gamma_1 (1 - \xi) \left(\frac{1}{s} - \delta \right)  + (1 - \eta_1) \gamma_1 \left(\frac{1}{s} - \delta \right)  \Bigr]
\end{align}
Dividing by $\gamma_1$ gives
$$
\gamma_1 \cdot \dot{\eta}_1 = \frac{\Gamma}{2} \left(\frac{1}{s} - \delta \right) \cdot \Bigl[ - (1 - \xi) + (1 - \eta_1) \Bigr] \implies \dot{\eta}_1 = \frac{\Gamma}{2} \left(\frac{1}{s} - \delta \right) \cdot \Bigl[ \xi - \eta_1 \Bigr]
$$
For $\dot{\gamma}_1$, we get the update rule
\begin{align}
\dot{\gamma}_1 & = - \gamma_1 \cdot \left(1 + \frac{1}{2s} - \frac{\delta}{2} - u\right) + \left(\frac{1}{2s} - \frac{\delta}{2}\right) \cdot \Gamma \notag \\
& = - \gamma_1 \cdot R + \left(\frac{1}{2s} - \frac{\delta}{2}\right) \cdot \Gamma
\end{align}
\end{proof}

\begin{proposition} \label{prop:monotone_phi}
The function $\Phi$ given by
\begin{align} \label{phi_fun_basic}
\Phi = \min\Bigl\{  \xi, \eta_1, \ldots, \eta_{8s} \Bigr\}
\end{align}
is weakly increasing.
\end{proposition}
\begin{proof}
Note that $\eta_j(0) = \xi(0) = \rho$.
We consider a few cases, depending on which term in the definition of $\Phi$ realizes the minimum at time $t$. We have a few cases:
\begin{itemize}
	\item If $\Phi(t) = \xi(t)$, then $\eta(t) \geq \xi(t)$ by equation (\ref{convex_y}), so $\dot{\xi}(t) \geq 0$.
	\item If $\Phi(t) = \eta_j(t)$, for some $j \in \{2, \ldots, 8s\}$, then  $\eta_{j-1}(t) \geq \eta_j(t)$. By definition of $\dot{\eta_j}(t)$ this implies $\dot{\eta_j}(t) \geq 0$.
	\item If $\Phi(t) = \eta_1(t)$, then $\xi(t) \geq \eta_1(t)$. By definition of $\dot{\eta_1}(t)$ this implies $\dot{\eta_1}(t) \geq 0$.
\end{itemize}

\medskip

It follows that the right derivative of $\Phi(t)$ is non-negative, so $\Phi$ is non-decreasing, since $\Phi$ is a continuous function (see page 208, solved exercise 19 in \cite{hardy_book}).

\end{proof}

\begin{proposition} \label{ub:delta}
The fraction of leaders with question mark is bounded by
$\delta(t) < 0.2/s$ for all $t$.
\end{proposition}
\begin{proof}
At the beginning of the protocol, $\delta(0) = 0$. Let $t$ be the minimum time at which $\delta(t) = 0.2/s$. We show that at this time, $\dot{\delta}(t) < 0$, and this will yield a contradiction.

First, rewrite $\dot{\delta}(t)$ as follows:
\begin{align} \label{delta_ub}
\dot{\delta}(t) & = \frac{\alpha(t)}{2} \cdot \left(\Gamma(t) - \beta(t) \right) + \frac{1/s - \alpha(t) - \delta(t)}{2} \cdot \beta(t) - \delta(t) \cdot \Gamma(t) \notag \\
& = \frac{\alpha(t)}{2} \cdot \left(\Gamma(t) - 2\beta(t) \right) + \frac{1/s - 0.2/s}{2} \cdot \beta(t) - \frac{0.2}{s} \cdot \Gamma(t) \notag \\
& = \Gamma(t) \cdot \left(\frac{\alpha(t)}{2} - \frac{0.2}{s}\right) - \left(\alpha(t) - \frac{0.4}{s}\right) \cdot \beta(t) \notag \\
& = \frac{\Gamma(t)- 2 \beta(t)}{2} \cdot \left(\alpha(t) - \frac{0.4}{s} \right)
\end{align}
Proposition~\ref{prop:monotone_phi} implies $\xi(t) \geq \rho$ and $\eta_j(t) \geq \rho$ at all times $t$. Then $\Gamma(t) - 2\beta(t) \geq 0$ and
\begin{align}
\xi(t) = \frac{1/s - \delta(t) - 2 \alpha(t)}{1/s - \delta(t)} = 1 -  \frac{2 \alpha(t)}{0.8/s} = 1 - \frac{s \cdot \alpha(t)}{0.4} \geq \rho \iff \alpha(t) - \frac{0.4}{s} \leq - \frac{0.4 \rho}{s} < 0
\end{align}
Replacing in (\ref{delta_ub}), we obtain $\dot{\delta}(t) < 0$. It follows that $\delta(t) < 0.2/s$ at all times $t$.
\end{proof}

\begin{lemma} \label{lem:R_lb}
Recall $R(t) = 1 + \frac{1}{2s} - \frac{\delta(t)}{2} - u(t)$.
If $s \geq 32$ and $u(t) \leq \frac{1}{se^2} \left(1 + \frac{3}{s}\right)$, then $R(t) \geq e^{\frac{1}{4s}}$.
\end{lemma}
\begin{proof}
Using the bound on $\delta(t)$ from Proposition~\ref{ub:delta} and the bound on $u(t)$ in the lemma statement, when $s \geq 32$ we get
\begin{align} 
R(t) & \geq 1 + \frac{1}{2s} - \frac{0.1}{s} - \frac{e^{-2}}{s} \cdot \left( 1 + \frac{3}{s} \right) \notag \\
& = 1 + \frac{0.4}{s} - \frac{e^{-2}}{s} \cdot \left(1 + \frac{3}{s} \right) \notag \\
& > 1 + \frac{1}{4s} + \frac{1}{16s^2} > e^{\frac{1}{4s}} \notag
\end{align}
The penultimate inequality is equivalent to $s > \frac{1/16 + 4 e^{-2}}{0.4 - e^{-2} - 0.25}$ and the last inequality follows since $e^x < 1 + x + x^2$ for $x < 1$.
\end{proof}

\begin{lemma} \label{ub:u_gamma}
Let $s \geq 32$ and $\epsilon \in (0, 1/2)$. If the following upper bounds hold for $\gamma_j$ and $u$ at time $t = t_0$, then they hold at all times $t \geq t_0$:
\begin{align}
\label{eq:ub_gammaj_phase1} \gamma_j(t) & < \frac{1 - \epsilon}{s} \cdot e^{-\frac{j}{4s}}, \; \forall j \in \{1, \ldots, 8s \}  \\
\label{eq:ub_u_phase2} u(t) & < \frac{1 - \epsilon}{s  e^{2}}  \left(1 + \frac{3}{s}\right)
\end{align}
\end{lemma}
\begin{proof}
The set of times $t$ at which of one or more of the inequalities in (\ref{eq:ub_gammaj_phase1}--\ref{eq:ub_u_phase2}) is violated is closed, so if this set is nonempty, it has a minimal element $t>t_0$.  We will derive a contradiction by establishing that the derivative of the first variable that violates the inequality at time $t$ is negative at that time. (The variables are ordered $\gamma_1,\ldots,\gamma_{8s},u$.)
	
\medskip

\noindent \emph{Case 1} : The first variable to violate the inequalities in (\ref{eq:ub_gammaj_phase1}--\ref{eq:ub_u_phase2}) is $\gamma_1$.
Using the bounds on $\delta(t)$ and $R(t)$, the derivative of $\gamma_1(t)$ at the time where $\gamma_1(t) = \frac{1-\epsilon}{s} \cdot e^{-\frac{1}{4s}}$ is
\begin{align}
\dot{\gamma}_1 & = - \gamma_1(t) \cdot R(t) + \left( \frac{1}{2s} - \frac{\delta(t)}{2} \right) \Gamma(t) \notag \\
& < - \frac{1-\epsilon}{s} \cdot e^{-\frac{1}{4s}} \cdot e^{\frac{1}{4s}} + \frac{1}{2s} < 0, \mbox{ since } \epsilon \in (0, 1/2)
\end{align}
Therefore if $\gamma_1(t)$ reaches the upper bound at time $t$, then its derivative at that time is strictly negative, which contradicts the definition of $t$ as the first time where the upper bound for $\gamma_1$ is reached.

\medskip

\noindent \emph{Case 2} : The first variable to violate the inequalities in (\ref{eq:ub_gammaj_phase1}--\ref{eq:ub_u_phase2}) is $\gamma_j$, for some $j \geq 2$.
Then
\begin{align} \label{ub_gamma_j_case1}
\dot{\gamma}_j(t) & = \gamma_{j-1}(t) - \gamma_j(t) \cdot R(t) \notag \\
& < \frac{1-\epsilon}{s} \cdot e^{-\frac{j-1}{4s}} - \frac{1-\epsilon}{s} \cdot e^{-\frac{j}{4s}} \cdot R(t) \notag \\
& = \frac{1-\epsilon}{s} \cdot e^{-\frac{j}{4s}} \cdot \Bigl[ e^{\frac{1}{4s}} - R(t) \Bigr]
\end{align}

Since (\ref{eq:ub_u_phase2}) holds before time $t$, the assumption of  Lemma~\ref{lem:R_lb} holds at time $t$, so $R(t) > e^{\frac{1}{4s}}$ when $s \geq 32$.
Thus $\dot{\gamma}_j(t) < 0$ by (\ref{ub_gamma_j_case1}), contradicting         the assumption of Case 2.

\medskip

\noindent \emph{Case 3} : The first variable to violate the inequalities in (\ref{eq:ub_gammaj_phase1}--\ref{eq:ub_u_phase2}) is $u$.
Then we get
\begin{align}
\dot{u}(t) & = \gamma_{8s}(t) - u(t) \cdot \Gamma(t) \notag \\
& < \frac{1-\epsilon}{se^{2}} \Bigl[ 1 - \left(1+\frac{3}{s}\right) \cdot \Gamma(t) \Bigr] \notag \\
& < \frac{1-\epsilon}{se^{2}} \Bigl[ 1 - \left(1+\frac{3}{s}\right)\cdot \left( 1 - \frac{2}{s} \right) \Bigr] \notag \\
& = \frac{1-\epsilon}{se^{2}} \Bigl[ 6/s^2-1/s  \Bigr] <0 \,, \mbox{ for } s > 6 \notag
\end{align}
This is in contradiction with the assumption of Case 3.

Thus there is no time $t$ at which one of the upper bounds in (\ref{eq:ub_gammaj_phase1}--\ref{eq:ub_u_phase2}) is violated.
\end{proof}

\begin{lemma} \label{gamma_ratio_lb}
For each $j \geq 2$, the ratio of consecutive $\gamma_j$'s is bounded as follows:
$$\frac{\gamma_{j-1}(t)}{\gamma_j(t)} > 1/2, \; \forall t \geq 0$$
\end{lemma}
\begin{proof}
	We use the method of first violation. Thus we look at the first time $t$ where the inequalities are violated; if multiple inequalities are violated at the same time $t$, then consider the smallest index $j$ with the property that $\gamma_{j-1}(t)/\gamma_j(t) = 1/2$.
If $j = 2$, then
\begin{align}
\left(\frac{\gamma_{1}(t)}{\gamma_{2}(t)} \right)' & = \frac{\gamma_2(t) \cdot \dot{\gamma}_1(t) - \gamma_1(t) \cdot \dot{\gamma}_2(t)}{\gamma_2(t)^2} \notag \\
& = \frac{\gamma_2(t) \cdot \Bigl(-\gamma_1(t) R(t) + \left(\frac{1}{2s} - \frac{\delta(t)}{2}\right)\Gamma(t)\Bigr) - \gamma_1(t) \Bigl(\gamma_1(t) - \gamma_2(t) R(t) \Bigr)}{\gamma_2(t)^2} \notag \\
& = \frac{\left(\frac{1}{2s} - \frac{\delta(t)}{2}\right)\Gamma(t)}{\gamma_2(t)} - \left( \frac{\gamma_1(t)}{\gamma_2(t)}\right)^2 \geq \frac{0.4}{s} \cdot \frac{\Gamma(t)}{\gamma_2(t)} - \left( \frac{\gamma_1(t)}{\gamma_2(t)}\right)^2 \notag \\
& \geq 0.4 \cdot \left(\frac{s-2}{s}\right) - \frac{1}{4} > 0
\end{align}
Thus the derivative of $\gamma_1/\gamma_2$ at time $t$ is strictly positive, which means that the value of $\gamma_1/\gamma_2$ couldn't have reached $1/2$ at time $t$. Since $t$ was the time of the first violation, this means there is no violation.

If
$j \geq 3$, then
\begin{align}
\left(\frac{\gamma_{j-1}(t)}{\gamma_j(t)} \right)' & = \frac{\gamma_{j}(t) \cdot \dot{\gamma}_{j-1}(t) - \gamma_{j-1}(t) \cdot \dot{\gamma}_j(t)}{\gamma_j(t)^2} \notag \\
& = \frac{\gamma_j(t)  \Bigl( \gamma_{j-2}(t) - \gamma_{j-1}(t) \cdot R(t)\Bigr)}{\gamma_j(t)^2} - \frac{\gamma_{j-1}(t)  \Bigl( \gamma_{j-1}(t) - \gamma_{j}(t) \cdot R(t)\Bigr)}{\gamma_j(t)^2}  \notag \\
& = \frac{\gamma_{j-1}(t)}{\gamma_j(t)} \cdot \Bigl[ \frac{\gamma_{j-2}(t)}{\gamma_{j-1}(t)} - \frac{\gamma_{j-1}(t)}{\gamma_j(t)} \Bigr]
\end{align}
When $\gamma_{j-1}(t) / \gamma_j(t) = 1/2$, we have
\begin{align}
\left(\frac{\gamma_{j-1}(t)}{\gamma_j(t)} \right)' & = \frac{\gamma_{j-1}(t)}{\gamma_j(t)} \cdot \Bigl[ \frac{\gamma_{j-2}(t)}{\gamma_{j-1}(t)} - \frac{\gamma_{j-1}(t)}{\gamma_j(t)} \Bigr] \notag \\
& = \frac{1}{2} \cdot \Bigl( \frac{\gamma_{j-2}(t)}{\gamma_{j-1}(t)} - \frac{1}{2} \Bigr) > 0
\end{align}
This means that the value of $\gamma_{j-1}/\gamma_j$ could not have dropped to $1/2$ at time $t$, in contradiction with $t$ being the time of the first violation.
It follows that $\gamma_{j-1}(t)/\gamma_j(t) > 1/2$ for all $j \geq 2$ at all times $t$.
\end{proof}

\medskip

By definition of $\Gamma$ and the upper bound on $u$ in (\ref{ub:u_gamma}),
\begin{align} \label{eq:Gamma_lb_ub}
1 - \frac{2}{s} < \Gamma(t) \leq 1 - \frac{1}{s}\,.
\end{align}

\begin{proposition} \label{prop:phase1_progress}
Let $\lambda \in (0, 1)$. There exists $C = C(\lambda, \rho) \leq \frac{144}{\rho (1 - \lambda)}$ and a time $T_{1} < C \cdot s$ so that $\xi(T_{1}) > \lambda$ and $\eta_j(T_{1}) > \lambda$ for all $j$. Consequently, at all times  $t \geq T_{1}$, we have $2 \beta(t) < 1 - \lambda$.
\end{proposition}
\begin{proof}
	From the expressions for the derivatives of $\eta_j$ and $\xi$, at a high level it can be observed that each $\eta_{j}$ follows $\eta_{j-1}$ when $j \geq 2$, $\eta_1$ follows $\xi$, and $\xi$ is pulled towards $\eta$ but in addition has a positive drift given by the term $\frac{\Gamma}{4} \eta (1 - \xi^2)$. This motivates the following definition:
	\begin{align} \label{psi_fun}
	\Psi(t) = \min\Bigl\{ \xi(t), \eta_1(t) + \epsilon \cdot v_1, \ldots, \eta_{8s}(t) + \epsilon \cdot v_{8s} \Bigr\},
	\end{align}
	where $\epsilon, v_1, \ldots, v_{8s}$ are free parameters that will be set later, where $0 < v_1 < \ldots < v_{8s}$. Note $\Psi(0) = \rho$.
	
	Let $\lambda_1 = (1 + \lambda)/2$ and $T_1 = T_1(\lambda) = \min\{t > 0 \; : \; \Psi(t) = \lambda_1\}$. The goal is to bound from below the right derivative $\Psi_{+}'(t)$ for $t < T_{1}$.
	
\medskip
	
	We consider a few cases, depending on which term in the definition of $\Psi$ realizes the minimum at time $t$:
	
	\smallskip
	
	\noindent \textbf{\emph{Case 1}}. The minimum is $\xi(t)$, that is, $\Psi(t) = \xi(t) < \lambda_1$. Then $\eta_j(t) + \epsilon v_j \geq \xi(t)$, so $\eta_j(t) \geq \xi(t) - \epsilon v_{8s}$ and $\eta(t) \geq \xi(t) - \epsilon v_{8s}$. Then
	\begin{align} \label{psi_plus_lb_case1}
	\Psi_{+}'(t) & \geq \dot{\xi}(t) = \frac{\delta(t)\Gamma(t)}{1/s - \delta(t)} \cdot \Bigl[ \eta(t) - \xi(t) \Bigr] + \frac{\Gamma(t)}{4} \left( 1 - \xi(t)^2 \right) \eta(t) \notag \\
	& \geq \frac{0.2/s}{0.8/s} \cdot (-\epsilon v_{8s}) + \frac{1}{8}(1 - \lambda_1^2)\rho
	\end{align}
	Let $\epsilon = \frac{\rho}{8}(1 - \lambda_1^2)$. Then for $v_{8s}=2$, the inequality in (\ref{psi_plus_lb_case1}) implies
	\begin{align}
	\Psi_{+}'(t) & \geq \epsilon \left( 1 - \frac{v_{8s}}{4}\right) = \frac{\epsilon}{2}
	\end{align}
	
	\noindent \textbf{\emph{Case 2}}. The minimum is $\Psi(t) = \eta_j(t) + \epsilon v_j$ for some $j \in \{2, \ldots, 8s\}$, so $\eta_j + \epsilon v_j \leq \eta_{j-1} + \epsilon v_{j-1}$.
	The right derivative of $\Psi$ is lower bounded by
	\begin{align}
	\Psi_{+}'(t) & \geq \dot{\eta}_j(t) = \frac{\gamma_{j-1}(t)}{\gamma_j(t)} \cdot \Bigl[ \eta_{j-1}(t) - \eta_j(t)\Bigr] \geq \frac{\epsilon}{2} \left( v_j - v_{j-1} \right)
	\end{align}	
	
	\noindent \textbf{\emph{Case 3}}. The minimum is $\Psi(t) = \eta_1(t) + \epsilon v_1$. Then the right derivative of $\Psi$ is lower bounded by
	\begin{align}
	\Psi_{+}'(t) & \geq \dot{\eta}_1(t) = \frac{\Gamma(t)}{2 \gamma_1(t)} \cdot \Bigl[ \frac{1}{s} - \delta(t) \Bigr] \left( \xi(t) - \eta_1(t) \right) \geq \frac{1 - 2/s}{2/s} \cdot \frac{0.8}{s} \cdot \epsilon v_1 \notag \\
	& \geq \frac{s}{4} \cdot \frac{0.8}{s} \cdot \epsilon v_1  = \frac{\epsilon v_1}{5}
	\end{align}
	
	By setting the $v_j$'s to satisfy the lower bounds with equality, i.e. $$
	\frac{v_1}{5} = \frac{v_2 - v_1}{2} = \ldots = \frac{v_{8s} - v_{8s-1}}{2}
	$$
we get $v_j  = \frac{2(3+ 2j)}{3+16s}$ for each $j \geq 1$.
\medskip
	
\noindent From cases $(1-3)$, it follows that at all times $t < T_{1}$, the right derivative $\Psi_{+}'(t) $ is large, i.e. it is bounded from below by $\epsilon/(9s)$. Then
	$$
	T_{1} \leq \frac{\lambda_1 - \rho}{(\epsilon/(9s))} \leq \frac{9s}{\epsilon}
	$$
	Moreover, at any time $t \geq T_{1}$,
	$$
	\min\{\xi(t), \eta_1(t), \ldots, \eta_{8s}(t)\} \geq \lambda_1 - 2 \epsilon > \lambda
	$$
	Taking
	$$
	C = C(\lambda, \rho) = 9/\epsilon = \frac{72}{\rho \cdot \left(1 - \left(\frac{1 + \lambda}{2}\right)^2\right)} \leq \frac{144}{\rho (1 - \lambda)}
	$$
	works. Moreover, since $y$ is a weighted average of $\eta_j$, for all $t > T_1$ we have
	$$\lambda < \eta(t) = \frac{\Gamma(t) - 2 \beta(t)}{\Gamma(t)} = 1 - \frac{2\beta(t)}{\Gamma(t)}\, .$$
	Therefore for all $t > T_1$ we have
	$
	2 \beta(t) < \Gamma(t)  \cdot (1 - \lambda) < 1 - \lambda
	$
as required.
\end{proof}

\section{Phases II and III: Exponential Decay of Errors} \label{sec:phases_2and3}

We divide the time after phase I in two phases, II and III. During phase II, we bound the fraction of wrong nodes by $\exp(-c_2 \cdot t/s)$ for some constant $c_2 > 0$. During phase III, we bound the fraction of wrong nodes by $\exp(-c_3 \cdot t)$ for some constant $c_3 > 0$.
Both of these will be handled in a similar way, by considering a potential function that is a linear combination of $\alpha$, $\delta$, and $\beta_j$ for $j \in \{1, \ldots, 8s+1\}$, with different coefficients for each of the two phases. The variable $\beta_{8s+1}(t)$ represents the fraction of nodes that are uninformed and have the wrong bit at time $t$. The derivative of $\beta_{8s+1}$ is
\begin{align} \label{eq:derivative_beta_8s1}
\dot{\beta}_{8s+1} = \beta_{8s} - \beta_{8s+1}\cdot \Gamma
\end{align}

The existence and uniqueness theorem (Corollary~\ref{cor:hurewicz}) applies when equation (\ref{eq:derivative_beta_8s1}) is added to the system of differential equations (\ref{alpha_dot}--\ref{gamma_1_dot}).

\subsection{Linear Potential Function}

Define
\begin{align} \label{def:Lambda}
\Lambda(t) = \alpha(t) + d \cdot \delta(t) + \sum_{j=1}^{8s+1} a_j \cdot \beta_j(t),
\end{align}
where $d, a_1, \ldots, a_{8s+1} \geq 0 $ have to be determined. We will require $\Lambda$ to satisfy the inequality $\dot{\Lambda}(t) \leq -\zeta \cdot \Lambda(t)$ for some constant $\zeta > 0$ and each time $t$ in the range of phases $II$ and $III$, respectively.

Bounding $\dot{\Lambda}$ using equations (\ref{alpha_dot}-\ref{eq:def_beta_1_dot}) and the inequality $\Gamma \leq 1$, we get
\begin{small}
\begin{align} 
\dot{\Lambda}(t) & \leq - \frac{\alpha(t)}{2} \cdot \Bigl(\Gamma(t) - \beta(t)\Bigr) + \delta(t) \beta(t) + d \cdot \frac{\alpha(t)}{2} \cdot \Gamma(t) + \frac{ d \beta(t)}{2s} - \frac{d \delta(t)}{2}  \notag \\
& \; \; \; \; + a_1  \Bigl[ -\beta_1(t)  R(t) + \frac{\alpha(t)}{2} \Bigr] + \sum_{j=2}^{8s} a_j  \Bigl[\beta_{j-1}(t) - \beta_j(t)  R(t) \Bigr] + a_{8s+1}   \Bigl[\beta_{8s}(t) - \beta_{8s+1}(t) \Gamma(t) \Bigr] \notag \\
& = \frac{\alpha(t)}{2}  \Bigl[ {\beta(t)} - 1 + \frac{2}{s} + {d} + {a_1} \Bigr] + \delta(t)  \Bigl[ \beta(t) - \frac{d}{2} \Bigr] + \sum_{j=1}^{8s} \beta_j(t)  \Bigl[\frac{d}{2s} - a_j  R(t) + a_{j+1} \Bigr] \notag
\end{align}
\end{small}

To ensure that $\dot{\Lambda}(t) \leq - \zeta \cdot \Lambda(t) = - \zeta \cdot \Bigl[ \alpha(t) + d \cdot \delta(t) + \sum_{j=1}^{8s} a_j \cdot \beta_j(t) \Bigr]$, it suffices that
\begin{align} \label{eq:phase2_constraints}
(a) \; \; \; & \frac{d}{2s} - a_j \cdot R(t) + a_{j+1} \leq - \zeta \cdot a_j, \; \forall j \leq 8s  \notag \\
& \iff a_{j+1} \leq a_j \cdot \Bigl(R(t) - \zeta \Bigr) - \frac{d}{2s}\, , \; \forall j \leq 8s \notag \\
(b) \; \; \; & \beta(t) - \frac{d}{2}  \leq - \zeta \cdot d  \iff \beta(t) \leq \left( \frac{1}{2} - \zeta \right) d \notag\\
(c) \; \; \; & \frac{\beta(t)}{2} - \frac{1 - 2/s}{2} + \frac{d}{2} + \frac{a_1}{2} \leq - \zeta \iff 2 \zeta + \beta(t) + d + a_1 \leq 1 - \frac{2}{s} \notag \\
(d) \; \; \; & - a_{8s+1} \cdot \Gamma \leq - a_{8s+1} \cdot \zeta
\end{align}

\subsection{Phase II} \label{sec:phase_2}

Let $\Lambda_{2}$ be the potential function for phase II, following the template in (\ref{def:Lambda}). The goal is to set the coefficients in $\Lambda_{2}$ so that  $\dot{\Lambda}_{2}(t) \leq -\zeta_{2} \cdot \Lambda_{2}(t)$, for some constant $\zeta_{2}> 0$.

Recall that at the end of Phase I we obtain that $\beta(t) \leq 1/64$ for $t \geq T_1$ when $\lambda = 31/32$ in Proposition~\ref{prop:phase1_progress}. Then there exists $C_1 > 0$ such that
\begin{align}  
T_1 \leq \frac{C_1 s}{\rho}, \mbox{ where } C_1 \leq 144 \cdot 32\,.
\end{align}
Also $R(t) \geq e^{\frac{1}{4s}} \geq 1 + \frac{1}{4s}$ from Lemma \ref{lem:R_lb} and $\Gamma(t) \geq 1 - 2/s$ from (\ref{eq:Gamma_lb_ub}).

\medskip

Set $a_j= a$, $\forall j \in \{1, \ldots, 8s\}$, and $a_{8s+1} = 0$, so part (d) of (\ref{eq:phase2_constraints}) holds.

For part (a) of (\ref{eq:phase2_constraints}) to hold, it suffices that $\frac{d}{2s} \leq a \left(\frac{1}{4s} - \zeta_{2} \right)$.
Let $\zeta_{2} = \frac{1}{8s}$ and $d = a/4$.

For part (b) of (\ref{eq:phase2_constraints}) to hold, it suffices that
$$
\beta(t) \leq \frac{d}{2} \cdot \left( 1 - \frac{1}{4s} \right)
$$
This holds when $\beta(t) \leq a/16$.

Finally, for part (c) of (\ref{eq:phase2_constraints}) to hold, it suffices that
$$
\beta(t) \leq \frac{1}{2} - \frac{5a}{4} - \frac{1}{4s}
$$
Let $a = 1/4$. Then the requirement on $\beta$ is $\beta(t) \leq 1/64$, which holds for all $t \geq T_1$.
Thus
\begin{align} \label{eq:lambda2_definition_exact_coeffs}
\Lambda_2(t) = \alpha(t) + \frac{\delta(t)}{16} + \frac{\beta(t)}{4}
\end{align}
It follows that $\dot{\Lambda}_{2}(t) \leq - \zeta_{2} \cdot \Lambda_{2}(t) = - \frac{\Lambda_{2}(t)}{8s}$, for all $t \geq T_1$.
\begin{align} \label{eq:lambda2_beyond_T1}
\Lambda_{2}(t) \leq \Lambda_{2}(T_1) \cdot \exp\left(-\frac{t - T_1}{8s}\right), \mbox{ for all } t \geq T_1\,.
\end{align}
Also note that $\Lambda_{2}(T_1) \leq 1$. We deduce that for $t \geq T_1$ and $j \leq 8s$
\begin{align}
a_j \cdot \beta_j(t) & = \frac{\beta_j(t)}{4} \leq \Lambda_{2}(t)
\implies \beta_j(t) \leq 4 \cdot \exp\Bigl(-\frac{t-T_1}{8s}\Bigr)
\end{align}
Thus $\beta$ can be bounded by
\begin{align}
\beta(t) = \sum_{j=1}^{8s} \beta_j(t) \leq  32s \cdot \exp\Bigl(-\frac{t-T_1}{8s}\Bigr)
\end{align}
In particular, if we take $T_2 = T_1 + 64 s^2$, then for all $t \geq T_2$
\begin{align} \label{eq:beta_ub_end_phase2}
\beta(t) \leq 32 s \cdot \exp\Bigl( -8s \Bigr)
\end{align}

\subsection{Phase III} \label{sec:phase_3}

Let $\Lambda_{3}$ be the potential function for phase III, following the template in (\ref{def:Lambda}). The goal is to set the coefficients in $\Lambda_{3}$ so that $\dot{\Lambda}_{3}(t) \leq -\zeta_{3} \cdot \Lambda_{3}(t)$, for some constant $\zeta_{3}> 0$.
The constraints needed to ensure that $\dot{{\Lambda}}_{3}(t) \leq - \zeta_{3} \cdot {\Lambda}_{3}(t)$ for $t \geq T_2$ are given in (\ref{eq:phase2_constraints}).

To satisfy part (a) of (\ref{eq:phase2_constraints}), set $\zeta_{3} = 1/2 - 2/s$ and ${a}_1 = 1/s$, and require that $\beta(t) + {d} \leq 1/s$. This also ensures that constraint $(d)$ is met.

\medskip

Part (b) of (\ref{eq:phase2_constraints}) is satisfied when $\beta(t) \leq 2 {d}/s$. Part (c) of (\ref{eq:phase2_constraints}) is satisfied when $ {d} \leq  {a}_j$ and $ {a}_{j+1} =  {a}_j/2$, since these imply
$$
{a}_{j+1} + \frac{ {d}}{2s} \leq  {a}_j \cdot \left( \frac{1}{2} + \frac{1}{2s} \right) =  {a}_j \cdot (1 - \zeta_{3})
$$
Thus let
\begin{align}
{a}_j & =  {a}_1 \cdot 2^{1-j} = \frac{1}{s} \cdot 2^{1-j} \notag \\
 {d} & = {a}_{8s} = \frac{1}{s} \cdot 2^{1-8s} \notag \\
\end{align}
These require that $\beta(t) \leq 2^{2-8s}/s^2$, which holds by (\ref{eq:beta_ub_end_phase2}) for all $t \geq T_2$.

\medskip

Recall $T_2 = O(s^2)$. For any $\theta > 0$, let $T_3 = T_2 + \frac{1 + \theta}{\zeta_{3}} \ln{n}$. Then $\Lambda_{3}(T_2) \leq 1$ implies
\begin{align}
\Lambda_{3}(T_3) \leq \exp \Bigl( -\zeta_{3} \cdot (T_3 - T_2) \Bigr) = n^{-1 - \theta}
\end{align}

Thus for all $j \leq 8s+1$,
$$\beta_j(T_3) \leq \Lambda_{3}(T_3) /  {a}_j \leq s \cdot  2^{j-1} n^{-1 - \theta}\,.$$

This implies that
\begin{align}
\alpha(T_3) + \beta(T_3) + \beta_{8s+1}(T_3) \leq n^{-1-\theta} \cdot \Bigl( 1 + s \cdot 2^{8s+1} \Bigr) = o\left(\frac{1}{n}\right)
\end{align}

\section{Random System}

Let $\widetilde{\alpha}(j) = $ number of leaders after $j$ rings in the finite system. Similarly, define variables $\widetilde{\delta}(j), \widetilde{\gamma}_1(j), \ldots, \widetilde{\gamma}_{8s}(j), \widetilde{\beta}_1(j), \ldots,  \widetilde{\beta}_{8s+1}(j)$. Also define
\begin{align}
\widetilde{\Gamma}(j) & = \sum_{i=1}^{8s} \widetilde{\gamma}_i(j) \notag \\
\widetilde{u}(j) & = n - \frac{n}{s} - \widetilde{\Gamma}(j) \notag \\
\widetilde{R}(j) & = n + \frac{n}{2s} - \frac{\widetilde{\delta}}{2} - \widetilde{u}(j) \notag
\end{align}

\medskip

Our goal is to show that the random system is closely approximated by the deterministic system analyzed in sections \ref{sec:deterministic_system_def}--\ref{sec:phases_2and3}. A key tool will be the following theorem by Warnke~\cite{warnke}, which refines earlier results by Kurtz~\cite{kurtz} and Wormald~\cite{wormald}. We include the main theorem in ~\cite{warnke} for the special case where the derivatives do not depend on time.

\begin{theorem}[\cite{warnke}]\label{thm:DEM}%
	Given integers~$a,n \ge 1$ and a bounded domain~$\mathcal{D} \subseteq \mathbb{R}^{a}$, let $F_k~:~\mathcal{D}~\to~\mathbb{R}$ be $L$-Lipschitz-continuous functions \footnote{With respect to the $\ell^{\infty}$ norm on $\mathbb{R}^a$.} on~$\mathcal{D}$ and $M=\sup\{F_k(y) \mid k \in [a], y \in \mathcal{D}\}$.
	Suppose $\{\mathcal{F}_i\}_{i \ge 0}$ are increasing $\sigma$-fields and
	$(Y_k(i))_{k=1}^{a}$ are $\mathcal{F}_i$-measurable random variables for all $i \ge 0$ with $Y_k(0)=n \cdot {y}_k^* $ where~$({y}_1^*, \ldots,  {y}_a^*) \in \mathcal{D}$. Moreover, assume that for all $i \geq 0$ and $1 \le k \le a$, the following conditions hold whenever~$\bigl(\frac{Y_1(i)}{n},...,\frac{Y_a(i)}{n}\bigr) \in \mathcal{D}$:
	\vspace{-0.25em}%
	\begin{enumerate}%
		\itemsep 0.125em \parskip 0em  \partopsep=0pt \parsep 0em
		\item[(i)]\label{dem:trend}%
		$\Bigl|\mathbb{E}\bigl[Y_k(i+1)-Y_k(i) \mid \mathcal{F}_{i}\bigr]-F_k\bigpar{\frac{Y_1(i)}{n},...,\frac{Y_a(i)}{n}} \Bigr| \le \lambda_0$,
		\item[(ii)]\label{dem:bounded}%
		$\bigabs{Y_k(i+1)-Y_k(i)}\le 1$. 
		 \end{enumerate}\vspace{-0.125em}%

	Let $T>0$ and $\lambda \ge \lambda_0 \min\{T,L^{-1}\} + M/n$. Let $(y_k(t))_{1 \le k \le a}$ be the unique solution for $t\le T$ of the system of differential equations
	\begin{equation}\label{dem:sol}
	y'_k(t) =F_k\bigpar{y_1(t), \ldots, y_a(t)} \quad \text{ with } \quad y_k(0) = {y}_k^* \qquad \text{for~$1 \le k \le a\,.$}
	\end{equation}
	 If $(y_1(t), \ldots y_a(t))$ has~$\ell^{\infty}$-distance at least~$3 e^{L T} \lambda$ from the boundary~of~$\mathcal{D}$ for all~$t \in [0,T]$, then
with probability at least $1-2a \cdot \exp(\frac{-n\lambda^2}{8T})$, we have
	\begin{equation}\label{dem:error}
	\max_{0 \le i \le T n} \max_{1 \le k \le a}\bigabs{Y_k(i)-n \cdot y_k\bigpar{\tfrac{i}{n}}} \; \leq \; 3 e^{L T} \lambda n \,.
	\end{equation}
\end{theorem}

\medskip

\begin{remark} \label{rmk:det_derivatives_map}
	Let $y = \left(\alpha, \beta_1, \ldots, \beta_{8s+1}, \gamma_1, \ldots, \gamma_{8s}, \delta\right) \in [0,1]^{16s+3}.$
	Recall $u = 1 - \frac{1}{s} - \sum_{j=1}^{8s} \gamma_j$, $\Gamma = \sum_{j=1}^{8s} \gamma_{j}$ and $R = 1 + \frac{1}{2s} - \frac{\delta}{2} - u.$
	Then Lemma~\ref{lem:def_deterministic} can be summarized by describing functions that give the time derivatives of the coordinates of $y$:
	\begin{align}
	& F_{\alpha}(y) = - \frac{\alpha}{2} \cdot \left(\Gamma - \beta \right) + \delta \cdot \beta \label{F_alpha} \\
	& F_{\delta}(y) = \frac{\alpha}{2} \cdot \left(\Gamma - \beta \right) + \frac{1/s - \alpha - \delta}{2} \cdot \beta - \delta \cdot \Gamma \\
	\label{F_beta_j}
	& F_{\beta_j}(y) = \beta_{j-1} - \beta_j \cdot R, \; \; \forall j \in\{2, \ldots, 8s\}\\
	\label{F_beta_1}
&	F_{\beta_1}(y) = -\beta_1 \cdot R + \frac{\alpha}{2} \cdot \Gamma \\
	\label{F_beta_8s1}
	& F_{\beta_{8s+1}}(y) = \beta_{8s} - \beta_{8s+1}\cdot \Gamma \\
	& F_{\gamma_j}(y) = \gamma_{j-1} - \gamma_j \cdot R, \; \; \forall j \in \{2, \ldots, 8s\} \label{F_gamma_j} \\
	& F_{\gamma_1}(y) = - \gamma_1 \cdot R + \left(\frac{1}{2s} - \frac{\delta}{2}\right) \cdot \Gamma \label{F_gamma_1}
	\end{align}
\end{remark}

In the discrete system, the derivatives are replaced by the differences in the values of variables at time $i+1$ and $i$, for each $i \in \mathbb{N}$.
Define $$\widetilde{{Y}}(i) = \left(\widetilde{\alpha}(i), \widetilde{\beta}_1(i), \ldots,  \widetilde{\beta}_{8s+1}(i),  \widetilde{\gamma}_1(i), \ldots, \widetilde{\gamma}_{8s}(i),\widetilde{\delta}(i)\right).$$

Let $\mathcal{F}_i$ be the history after $i$ clock rings for the process $\widetilde{Y}$. The next lemma verifies condition (i) from Theorem~\ref{thm:DEM}.

\begin{lemma} \label{lem:change_expectation_Y}
The expected change in $\widetilde{\alpha}$ is approximated by
\begin{align} \label{eq:exp_change_F_alpha}
\Bigl|\Ex \left[\widetilde{\alpha}(i+1) - \widetilde{\alpha}(i) \mid \mathcal{F}_i \right]  - F_{\alpha}({\widetilde{Y}(i)}/{n}) \Bigr| \leq \frac{1}{n}
\end{align}
More generally, for each $k = 1, \ldots 16s+3$,
\begin{align} \label{eq:expected_change_Y_tilde}
\Bigl| \Ex \left[\widetilde{Y}_k(i+1) - \widetilde{Y}_k(i) \mid \mathcal{F}_i \right]  - F_{{y}_k}({\widetilde{Y}(i)}/{n}) \Bigr| \leq \frac{1}{n}
\end{align}
\end{lemma}
\begin{proof}
The expected change in the number of leaders with the incorrect bit is given by
\begin{align} \label{eq:alpha_1step}
\Ex\left[\widetilde{\alpha}(i+1) - \widetilde{\alpha}(i) \mid \mathcal{F}_i \right] & = - \frac{\widetilde{\alpha}(i)}{2n} \cdot \frac{\widetilde{\Gamma}(i)-\widetilde{\beta}(i)}{n-1} + \frac{\widetilde{\delta}(i)}{n} \cdot \frac{\widetilde{\beta}(i)}{n-1}
\end{align}
The first term in the update (\ref{eq:alpha_1step}) represents the probability that the node selected at step $i+1$ is a leader with the incorrect bit that pulls from an informed node with the correct bit. The second term is the probability that the node selected is a leader with \q that pulls from an informed node with the wrong bit.

\medskip

On the other hand, by Remark~\ref{rmk:det_derivatives_map}
\begin{align} \label{eq:alpha_1step_Ytilde}
F_{\alpha}(\widetilde{Y}(i)/n) = - \frac{\widetilde{\alpha}(i)}{2n} \cdot \frac{\widetilde{\Gamma}(i) - \widetilde{\beta}(i)}{n} + \frac{\widetilde{\delta}(i)}{n} \cdot \frac{\widetilde{\beta}(i)}{n}
\end{align}
Subtracting (\ref{eq:alpha_1step_Ytilde}) from (\ref{eq:alpha_1step}) implies (\ref{eq:exp_change_F_alpha}).


The expected change in the number of leaders with \q is given by
\begin{align}
\Ex\left[\widetilde{\delta}(i+1) - \widetilde{\delta}(i) \mid \mathcal{F}_i \right] & = \frac{\widetilde{\alpha}(i)}{2n} \cdot \frac{\widetilde{\Gamma}(i)-\widetilde{\beta}(i)}{n-1} + \frac{n/s - {\widetilde{\alpha}(i)} - {\widetilde{\delta}(i)}}{2n} \cdot \frac{\widetilde{\beta}(i)}{n-1} - \frac{\widetilde{\delta}(i)}{n} \cdot \frac{\widetilde{\Gamma}(i)}{n-1} \notag
\end{align}
The first term is the probability that a wrong leader pulled from a correctly informed node, the second term is the probability that a correct leader pulled from an incorrectly informed node, and the last term is the probability that a leader with \q pulled from an informed node.

\medskip

On the other hand, by Remark~\ref{rmk:det_derivatives_map} and Lemma~\ref{lem:def_deterministic}
\begin{align}
F_{\delta}(\widetilde{Y}(i)/n) = \frac{\widetilde{\alpha}(i)}{2n} \cdot \frac{\widetilde{\Gamma}(i)-\widetilde{\beta}(i)}{n}  + \frac{n/s - {\widetilde{\alpha}(i)} - {\widetilde{\delta}(i)}}{2n} \cdot \frac{\widetilde{\beta}(i)}{n} - \frac{\widetilde{\delta}(i)}{n} \cdot \frac{\widetilde{\Gamma}(i)}{n} \notag
\end{align}
Subtracting the last two displayed equations yields (\ref{eq:expected_change_Y_tilde}) for $\delta$.

The expected change in the number of informed nodes in bin $1$ with the wrong bit is
\begin{small}
	\begin{align}
	\Ex\left[\widetilde{\beta}_1(i+1) - \widetilde{\beta}_1(i) \mid \mathcal{F}_i \right] & = - \frac{\widetilde{\beta}_1(i)}{n}- \frac{n/s - \widetilde{\delta}(i) - \widetilde{\alpha}(i)}{2n} \cdot \frac{\widetilde{\beta}_1(i)}{n-1} + \frac{\widetilde{\alpha}(i)}{2n} \cdot \frac{\widetilde{\Gamma}(i) - \widetilde{\beta}_1(i)}{n-1}  + \frac{\widetilde{u}(i)}{n} \cdot \frac{\widetilde{\beta}_1(i)}{n-1} \notag
	\end{align}
\end{small}
The first term represents the probability the selected node is a follower with counter value $1$ that increases its counter value, thus leaving bin $1$. The second term is the probability the selected node is a follower that gets pushed a correct bit from the a leader (and so leaves the set of followers with the wrong bit and counter value 1). The third term is the probability the selected node is a follower that gets pushed the wrong bit from a leader with incorrect information, and the last term is the probability that an uninformed follower pulls the wrong bit from a follower with counter value 1. Comparing this to $F_{\beta_1}(\widetilde{Y}(i)/n)$ using Lemma~\ref{lem:def_deterministic} yields (\ref{eq:expected_change_Y_tilde}) for $\beta_1$.

\medskip

The expected change in the number of informed nodes in bin $j \in \{2, \ldots, 8s\}$ with the wrong bit is
\begin{small}
\begin{align}
\Ex\left[\widetilde{\beta}_j(i+1) - \widetilde{\beta}_j(i) \mid \mathcal{F}_i \right] & =
 \frac{\widetilde{\beta}_{j-1}(i)}{n} -  \frac{\widetilde{\beta}_{j}(i)}{n}- \frac{n/s - \widetilde{\delta}(i)}{2n} \cdot \frac{\widetilde{\beta}_j(i)}{n-1} + \frac{\widetilde{u}(i)}{n} \cdot \frac{\widetilde{\beta}_j(i)}{n-1}, \; \; \forall j \in\{2, \ldots, 8s\} \notag
 \end{align}
 \end{small}
 The first and second term represent the probability the selected node is a wrongly informed node from bin $j-1$ and bin $j$, respectively. The third term is the probability the selected node is a leader that pushes its bit to a wrong node in bin $j$, and the last term is the probability that an uniformed node pulls from an incorrect node in bin $j$.
 Comparing this to $F_{\beta_{j}}(\widetilde{Y}(i)/n)$ using Lemma~\ref{lem:def_deterministic} yields (\ref{eq:expected_change_Y_tilde}) for $\beta_j$.

\medskip

The expected change in the number of uninformed nodes with the wrong bit is 
\begin{small}
	\begin{align}
	\Ex\left[\widetilde{\beta}_{8s+1}(i+1) - \widetilde{\beta}_{8s+1}(i) \mid \mathcal{F}_i \right] & =
	\frac{\widetilde{\beta}_{8s}(i)}{n} -  \frac{\widetilde{\beta}_{8s+1}(i)}{n} \cdot \frac{\widetilde{\Gamma}(i)}{n-1} \notag
	\end{align}
\end{small}
The first term represents the probability the selected node is a wrongly informed node from bin $8s$ and the second is the probability that an uninformed node with the wrong bit is selected and contacts an informed node.
Comparing this to $F_{\beta_{8s+1}}(\widetilde{Y}(i)/n)$ using (\ref{eq:derivative_beta_8s1}) yields (\ref{eq:expected_change_Y_tilde}) for $\beta_{8s+1}$.

\medskip


The arguments for $\gamma_j$ are similar to those for $\beta_j$, so we only include the expected changes:
\begin{small}
\begin{align}
\Ex\left[\widetilde{\gamma}_j(i+1) - \widetilde{\gamma}_j(i) \mid \mathcal{F}_i \right] & = \frac{\widetilde{\gamma}_{j-1}(i)}{n} - \frac{\widetilde{\gamma}_j(i)}{n} - \frac{\widetilde{\gamma}_j(i)}{n-1} \cdot \Bigl( \frac{1}{2s} -  \frac{\widetilde{\delta}(i)}{2n} - \frac{\widetilde{u}(i)}{n}\Bigr) \; \; \forall j \in\{2, \ldots, 8s\} \notag \\
\Ex\left[\widetilde{\gamma}_1(i+1) - \widetilde{\gamma}_1(i) \mid \mathcal{F}_i \right] & = - \frac{\widetilde{\gamma}_1(i)}{n} + \left(\frac{1}{2s} - \frac{\widetilde{\delta}(i)}{2n} \right) \cdot  \frac{\widetilde{\Gamma}(i)-\widetilde{\gamma}_1(i)}{n-1} + \frac{\widetilde{u}(i)}{n} \cdot \frac{\widetilde{\gamma}_1(i)}{n-1} \notag
\end{align}
\end{small}
\end{proof}

\begin{corollary} \label{cor:warnke_bounds_beta_j}
	Let $ T \leq \frac{\ln{n}}{240s}$.
	Suppose $\alpha(0) = \widetilde{\alpha}(0) / n $ and similarly for the other variables. There exists an event $A$ where $\mathbb{P}(A) \geq 1 - e^{-2n^{1/3}}$, such that on $A$ the following inequalities hold:
	\begin{align}
	& \max_{0 \leq i \leq Tn} \mid \widetilde{\alpha}(i) - n \cdot \alpha\left({i}/{n}\right)| \leq 3 n^{7/8} \notag \\
	& \max_{0 \leq i \leq Tn} \mid \widetilde{\beta}_j(i) - n \cdot \beta_j\left({i}/{n}\right)| \leq 3 n^{7/8}, \mbox{ for each } j \in \{1, \ldots, 8s+1\}\notag  \\
	& \max_{0 \leq i \leq Tn} \mid \widetilde{\delta}(i) - n \cdot \delta\left({i}/{n}\right)| \leq 3 n^{7/8} \notag \\
	& \max_{0 \leq i \leq Tn} \mid \widetilde{\gamma}_j(i) - n \cdot \gamma_j\left({i}/{n}\right)| \leq 3 n^{7/8}, \mbox{ for each } j \in \{1, \ldots, 8s\} \notag
	\end{align}
\end{corollary}
\begin{proof}
	We apply the differential equation method (Theorem~\ref{thm:DEM}). The variables are defined in Remark~\ref{rmk:det_derivatives_map}.
Let $a = 16s+3$ and consider the following domain:
\begin{align}
\mathcal{D}  = \Bigl\{ \bigl(\alpha, \delta, \beta_1, \ldots, \beta_{8s+1}, \gamma_1, \ldots, \gamma_{8s} \bigr) \mid
\alpha, \delta, \beta_j, \gamma_j \in [-1, 1] \; \forall j,  \sum_{j=1}^{8s+1} \beta_j \in (-1,1), \sum_{j=1}^{8s} \gamma_j \in (-1,1)  \Bigr\}
\end{align}
Examining the functions in Remark~\ref{rmk:det_derivatives_map}, it can be seen that $M = \sup\{F_k(y) \mid k \in [a], y \in \mathcal{D}\} \leq 4$ and the Lipschitz constant is $L = 30s$.

By Lemma~\ref{lem:change_expectation_Y}, condition (i) of Theorem~\ref{thm:DEM} holds with $\lambda_0 = 1/n$. Let $\lambda = n^{-1/4}$. Then by definition of $\mathcal{D}$ and Lemma~\ref{lem:bounded_differential}, the distance from the boundary is at least $1/s$ for all $t \in [0, T]$. Since
$
T \leq {\ln{n}}/(240s) 
$
we get that $3e^{LT}\lambda = 3e^{30s \cdot T} \cdot n^{-1/4} \leq 3 e^{\ln\left({n^{1/8}}\right)} n^{-1/4} \leq 3 n^{-1/8} < 1/s$.
Then with probability at least $1 - \exp\left(-2n^{1/3}\right)$, we have
\begin{align}
\max_{0 \le i \le T n} \max_{1 \le k \le a}\bigabs{Y_k(i)-n \cdot y_k\bigpar{\tfrac{i}{n}}} \; \leq \; 3 n^{7/8} \,.
\end{align}
\end{proof}

\begin{proposition} \label{prop:warnke_app}
For each $i \in \{n T_2, \ldots, 4n \ln{n}\}$, let $G_i$ be the event where the following constraints hold:
	\begin{align}
	& \widetilde{\beta}(i) \leq \frac{n}{s^2} \cdot {2^{2-8s}} \\
	& \frac{\widetilde{\delta}(i)}{2} + \widetilde{u}(i) \leq \frac{n-1}{2s}
	\end{align}
Define
$G = \bigcap_{i = nT_2}^{4n \ln{n}} G_i\,.$
Then $\Pr(G^c) \leq e^{-n^{1/3}}$.
\end{proposition}
\begin{proof}
Recall $T_2 = T_1 + 64 s^2 \leq \frac{C_1 s}{\rho} + 64 s^2$, where $C_1 \leq 144 \cdot 32$.
We apply the differential equation method iteratively on time intervals
\begin{align}
\{\ell nT_2, \ldots, (\ell+1)nT_2\}, \mbox{ for } \ell \in \{0, 1,\ldots, \ln{n}\}\,.
\end{align}
At the beginning of each $\ell$-th iteration, we reset the deterministic system to match the random one; we denote the $\ell$-th deterministic system set this way by $\alpha^{[\ell]}(t)$, where $t \in \{\ell T_2, \ldots, (\ell +1)T_2\}$. The initial condition is $\alpha^{[\ell ]}(\ell T_2) := \widetilde{\alpha}(\ell nT_2)/n$. The other variables are similarly set. Note that $\alpha^{[0]} = \alpha$. An example can be seen in the next figure.

\begin{figure}[h!]
	\centering
	\subfigure[Fraction of leaders with the wrong bit]
	{
		\includegraphics[scale=1.52]{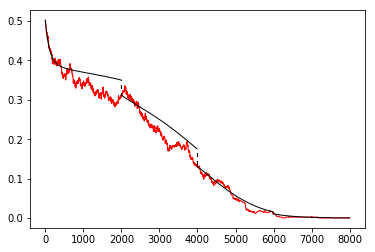}
	}
	\subfigure[Fraction of informed followers with the wrong bit]
	{
		\includegraphics[scale=1.52]{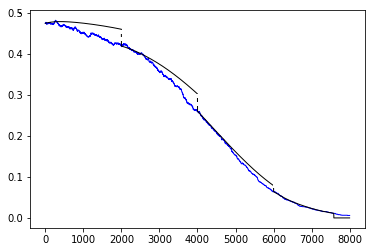}
	}
	\caption{The fraction of leaders with the wrong bit, where the deterministic system is reset periodically to restart from the random system ($\frac{\widetilde{\alpha}}{n} \cdot s$ shown in red and $\frac{\widetilde{\beta}}{n} \cdot \frac{s}{s-1}$ in blue). The $X$ axis shows the time, $n = 2500$, $s=5$, and the initial minority fraction is $0.49$}
\end{figure}

By Corollary~\ref{cor:warnke_bounds_beta_j}, there exists an event $A_{\ell}$ with $\mathbb{P}(A_{\ell}) \geq 1 - e^{-2n^{1/3}}$, so that on $A_{\ell}$ the following inequalities hold at all times $i \in \{\ell nT_2, \ldots, (\ell +1)nT_2\}$:
\begin{align} \label{eq:beta_close_beta_tilde}
& \mid \widetilde{\alpha}(i) - n \cdot \alpha^{[\ell]}(i/n) \mid \leq 3 n^{7/8} \notag \\
& \mid \widetilde{\delta}(i) - n \cdot \delta^{[\ell]}(i/n) \mid \leq 3 n^{7/8} \notag \\
& \mid \widetilde{\beta}(i) - n \cdot \beta^{[\ell]}(i/n) \mid \leq 24s \cdot n^{7/8} \notag \\
& \mid \widetilde{\Gamma}(i) - n \cdot \Gamma^{[\ell]}(i/n) \mid \leq 24s \cdot n^{7/8}
\end{align}
Recall $\Lambda_2(t) = \alpha(t) + \delta(t)/16 + \beta(t)/4$.
Summing up the inequalities in (\ref{eq:beta_close_beta_tilde}) with suitable weights, we get that on the event $A_{\ell}$ the next inequality holds
\begin{align} \label{eq:tilde_lambda_deterministic_interval_bound}
\mid \widetilde{\Lambda}_2(i) - n \cdot \Lambda_2^{[\ell]}(i/n) \mid \leq 7s \cdot n^{7/8}
\end{align}
The proof of inequality (\ref{eq:lambda2_beyond_T1}) gives
\begin{align} \label{eq:lambda2_beyond_T1_kth_system}
\Lambda_{2}^{[\ell]}(t) \leq \Lambda_{2}^{[\ell ]}(\ell T_2) \cdot \exp\left(-\frac{t - \ell T_2}{8s}\right), \mbox{ for all } t \geq \ell T_2\,.
\end{align}
By (\ref{eq:lambda2_beyond_T1}), it follows that $\Lambda_{2}^{[0]}(T_2) \leq e^{-8s}$. Therefore, on the event $A_0$, we have
\begin{align} \label{eq:lambda2_base_case_k=1}
\Lambda_{2}^{[1]}(T_2) = \frac{\widetilde{\Lambda}_2(nT_2)}{n} \leq \exp(-8s) + 7s \cdot n^{-1/8}\,.
\end{align}
Let $A^* = \bigcap_{\ell =0}^{\ln(n)} A_{\ell }$. Note that $\Pr((A^*)^c) \leq (1 + \ln{n}) \cdot e^{-2n^{1/3}} \leq e^{-n^{1/3}}$.
We use induction on $\ell  = 1, \ldots, \ln{n}$ to show that on $A^*$ we have
\begin{align} \label{eq:lambda2_k_induction_bound}
\Lambda_2^{[\ell ]}(t) \leq \exp(-8s) + 7s \cdot n^{-1/8}, \mbox{ for all } \ell  = 1, \ldots, \ln{n} \mbox{ and } t \in [\ell T_2, (\ell +1)T_2]
\end{align}
The base case $\ell =1$ holds by (\ref{eq:lambda2_base_case_k=1}). We assume (\ref{eq:lambda2_k_induction_bound}) holds for $\ell $ and derive it for $\ell +1$. By applying (\ref{eq:lambda2_beyond_T1_kth_system}), we obtain
\begin{align}  \label{eq:lambda2_k+1_induction_bound_initialTk}
\Lambda_2^{[\ell ]}((\ell +1) T_2) \leq \Bigl(\exp(-8s) + 7s \cdot n^{-1/8} \Bigr) \cdot \exp(-8s) \leq \exp(-8s)
\end{align}
Then on $A^*$ we have
\begin{align}\label{eq:lambda2_k+1_induction_bound_initialTk+1}  \Lambda_2^{[\ell +1]}((\ell +1)T_2) & = \frac{\widetilde{\Lambda}_2((\ell +1)nT_2)}{n} \notag \\
& \leq \Lambda_2^{[\ell ]}((\ell +1) T_2) + 7s \cdot n^{-1/8} \notag \\
& \leq \exp(-8s) + 7s \cdot n^{-1/8}
\end{align}
The induction claim follows from (\ref{eq:lambda2_beyond_T1_kth_system}) and (\ref{eq:lambda2_k+1_induction_bound_initialTk+1}). Note that $T_2 \geq 4$, thus for each $i \in \{nT_2, \ldots, 4n\ln{n}\}$, on $A^*$ we have
\begin{align} \label{eq:beta_delta_tilde_i_ub}
 & \widetilde{\beta}(i) \leq 4\left(\exp(-8s) + 7s \cdot n^{-1/8}\right) n \leq \frac{n}{s^2} \cdot 2^{2-8s}  \notag \\
& \widetilde{\delta}(i) \leq 16\left(\exp(-8s) + 7s \cdot n^{-1/8}\right) n\,.
\end{align}

Let $\epsilon_0 > \epsilon_1 > \ldots > \epsilon_{\;\ln{n}} > 0$, where $\epsilon_0 = 1/50$. For $\ell  > 0$, define $\epsilon_{\ell}= \epsilon_{\ell -1} - 1/(60 \ln{n})$. We will show by induction on $\ell = 0, \ldots, \ln{n}$ that on $A^*$, the following inequalities hold
\begin{align}
\gamma_j^{[\ell ]}(t) & < \frac{1 - \epsilon_{\ell }}{s} \cdot e^{-\frac{j}{4s}}, \; \forall j \in \{1, \ldots, 8s \}, \forall t \geq \ell T_2 \\
u^{[\ell ]}(t) & < \frac{1 - \epsilon_{\ell }}{s  e^{2}}  \left(1 + \frac{3}{s}\right), \forall t \geq \ell T_2
\end{align}
The base case $\ell =0$ follows by Lemma~\ref{ub:u_gamma} since $\gamma_j(0) = 1/(8s)$ and $u(0) = 0$. Assume it holds for $\ell -1$ and derive it for $\ell $. Recall $u(t) = 1 - 1/s - \Gamma(t)$. Then by (\ref{eq:beta_close_beta_tilde}), on $A^*$ we have
\begin{align}
\Bigl | u^{[\ell ]}(\ell T_2) - u^{[\ell -1]}(\ell T_2) \Bigr |  = \Bigl | \frac{\widetilde{u}(n \ell T_2)}{n} - u^{[\ell -1]}(\ell T_2) \Bigr | \leq 24s n^{-1/8} \le \frac{\epsilon_{\ell } - \epsilon_{\ell -1}}{se^2},
\end{align}
and similarly for $\gamma_j^{[\ell ]}(\ell T_2)$.
Applying Lemma~\ref{ub:u_gamma} again completes the induction step.
By (\ref{eq:beta_close_beta_tilde}), for each $i \in \{\ell nT_2, \ldots, (\ell +1)nT_2\}$, on $A^*$ we have
\begin{align} \label{eq:u_tilde_i_ub}
\widetilde{u}(i) \leq n \cdot u^{[\ell ]}(i/n) + 24s \cdot n^{7/8} \leq \frac{n}{se^2} \left(1 + \frac{3}{s} \right) + 24s \cdot n^{7/8}
\end{align}
Combining (\ref{eq:beta_delta_tilde_i_ub}) and (\ref{eq:u_tilde_i_ub}) implies that  $A^* \subseteq G$. Since $\Pr(A^*) \geq 1- e^{-n^{1/3}}$, this completes the argument.
\end{proof}

\begin{proposition} \label{prop:exp_lambda_3_tilde}
For each integer $i \geq 0$, let $\widetilde{\Lambda}_{3}(i) = \widetilde{\alpha}(i) + d \cdot \widetilde{\delta}(i) + \sum_{j=1}^{8s+1} a_j \cdot \widetilde{\beta}_j(i)$, where $d =1/s \cdot 2^{1-8s}$ and $a_j = 1/s \cdot 2^{1-j}$ for all $j$. Then for each $i \geq n T_2$
$$
\Ex\Bigl[\widetilde{\Lambda}_{3}(i+1) - \widetilde{\Lambda}_{3}(i) \Bigr] \leq - \frac{\zeta_{3}}{n} \cdot \Ex\left[\widetilde{\Lambda}_{3}(i)\right] + e^{-n^{1/3}}, \quad \mbox{for } \zeta_{3} = 1/2 - 2/s\,.
$$
\end{proposition}
\begin{proof}
By Lemma~\ref{lem:change_expectation_Y}, the change in the number of leaders with the wrong bit is
\begin{align} \label{eq:tilde_alpha_ub}
\Ex\left[ \widetilde{\alpha}(i+1) - \widetilde{\alpha}(i) \mid \mathcal{F}_i \right] & = \frac{\widetilde{\alpha}(i)}{2n(n-1)} \cdot \left[{-\widetilde{\Gamma}(i) + \widetilde{\beta}(i)}\right] + \frac{\widetilde{\delta}(i)\widetilde{\beta}(i)}{n(n-1)}
\end{align}
Recall the definition of $G_i$ from Proposition~\ref{prop:warnke_app}.
On $G_i$, since $\widetilde{\Gamma}(i) > n/2$,   Lemma~\ref{lem:change_expectation_Y} gives the following bound on the number of \q nodes:
\begin{align}  \label{eq:tilde_delta_ub}
\Ex\left[ \widetilde{\delta}(i+1) - \widetilde{\delta}(i) \mid \mathcal{F}_i \right] & \leq \frac{\widetilde{\alpha}(i)}{2n(n-1)} \cdot \left[{\widetilde{\Gamma}(i) - \widetilde{\beta}(i)}\right] + \frac{\widetilde{\beta}(i)}{2s(n-1)}  - \frac{\widetilde{\delta}(i)}{2(n-1)}
\end{align}

For the number of informed nodes in bin $1$ with the wrong bit, Lemma~\ref{lem:change_expectation_Y} gives the upper bound:
\begin{align} \label{eq:tilde_beta_1_ub}
\Ex\left[ \widetilde{\beta}_1(i+1) - \widetilde{\beta}_1(i) \mid \mathcal{F}_i \right] & \leq -\frac{\widetilde{\beta}_1(i)}{n} - \frac{\widetilde{\beta}_1(i)}{2s(n-1)}  + \frac{\widetilde{\alpha}(i)\widetilde{\Gamma}(i)}{2n(n-1)} + \frac{\widetilde{u}(i)\widetilde{\beta}_1(i)}{n(n-1)}
\end{align}

On $G_i$, since $\widetilde{\delta}(i) + 2 \widetilde{u}(i) \leq \frac{n-1}{2s}$, Lemma~\ref{lem:change_expectation_Y} gives the following upper bound on the number of informed nodes in bin $j \in\{2, \ldots, 8s\}$ with the wrong bit:
\begin{align} \label{eq:tilde_beta_j_ub}
\Ex\left[ \widetilde{\beta}_j(i+1) - \widetilde{\beta}_j(i) \mid \mathcal{F}_i \right] & \leq \frac{\widetilde{\beta}_{j-1}(i) - \widetilde{\beta}_{j}(i)}{n} -  \frac{\widetilde{\beta}_j(i)}{2n(n-1)} \cdot \frac{n}{2s}
\end{align}

Finally, for the number of uninformed nodes with the wrong bit, Lemma \ref{lem:change_expectation_Y} gives the bound:
\begin{align} \label{eq:tilde_beta_8s1_ub}
\Ex\left[ \widetilde{\beta}_{8s+1}(i+1) - \widetilde{\beta}_{8s+1}(i) \mid \mathcal{F}_i \right] & \leq \frac{\widetilde{\beta}_{8s}(i)}{n} -  \frac{\widetilde{\beta}_{8s+1}(i)}{n} \cdot \frac{\widetilde{\Gamma}(i)}{n-1}
\end{align}

By adding inequality (\ref{eq:tilde_alpha_ub}), inequality (\ref{eq:tilde_delta_ub}) multiplied by $d$, inequality (\ref{eq:tilde_beta_1_ub}) multiplied by $a_1$, inequality (\ref{eq:tilde_beta_j_ub}) multiplied by $a_j$, and inequality (\ref{eq:tilde_beta_8s1_ub}) multiplied by $a_{8s+1}$, we obtain that on $G_i$ the conditional expectation of $\widetilde{\Lambda}_3$ can be bounded by:
\begin{small}
\begin{align} \label{eq:lambda_3_tilde_ub}
\mathbbm{1}_{G_i} \cdot \Ex\left[ \widetilde{\Lambda}_{3}(i+1) - \widetilde{\Lambda}_{3}(i) \mid \mathcal{F}_i \right] & \leq \mathbbm{1}_{G_i} \cdot \Bigl\{ \frac{\widetilde{\alpha}(i)}{2n(n-1)}  \left[\widetilde{\Gamma}(i)  (d+a_1 - 1) + \widetilde{\beta}(i)  (1 - d) \right] + \frac{\widetilde{\delta}(i)}{n(n-1)} \left( \widetilde{\beta}(i) - d  \frac{n}{2} \right) \Bigr. \notag \\
&  \Bigl. + \sum_{j=1}^{8s} \frac{\widetilde{\beta}_j(i)}{n} \left[ a_{j+1} - a_j \left( 1 + \frac{1}{2s} \right) + \frac{dn}{2s(n-1)} \right] - \frac{\widetilde{\beta}_{8s+1}(i)}{2n} \Bigr\}
\end{align}
\end{small}

Using $\widetilde{\Gamma}(i) \geq (1 - \frac{2}{s})n$ and $d + a_1 < 2/s$, we get that on $G_i$ the conditional expectation can further be bounded as follows:
\begin{small}
\begin{align}
\mathbbm{1}_{G_i} \cdot \Ex\left[ \widetilde{\Lambda}_{3}(i+1) - \widetilde{\Lambda}_{3}(i) \mid \mathcal{F}_i \right] & \leq \mathbbm{1}_{G_i} \cdot
\Bigl\{ \frac{\widetilde{\alpha}(i)}{2n} \left[ - \left(1 - \frac{2}{s}\right)^2 + \frac{1}{s^2}\right] +  \frac{\widetilde{\delta}(i)}{n} \left[ - \frac{d}{2}\left(1 - \frac{4}{s}\right) \right] \Bigr. \notag \\
& \Bigl. + \sum_{j=1}^{8s+1} \frac{\widetilde{\beta}_j(i)}{n} \cdot \left[ - \frac{a_j}{2} \left( 1 - \frac{1}{s} \right) \right] \Bigr\} \notag \\
& \leq - \frac{\zeta_3}{n} \cdot \widetilde{\Lambda}_3(i) \notag
\end{align}
\end{small}
Taking expectation over all histories $\mathcal{F}_i$, we get the following inequality:
\begin{small}
\begin{align}
\Ex\left[ \widetilde{\Lambda}_{3}(i+1) - \widetilde{\Lambda}_{3}(i) \right] & = \mathbb{P}(G_i) \Ex\left[\widetilde{\Lambda}_{3}(i+1) - \widetilde{\Lambda}_{3}(i) \mid G_i \right] + \mathbb{P}(G_i^c) \Ex\left[\widetilde{\Lambda}_{3}(i+1) - \widetilde{\Lambda}_{3}(i) \mid G_i^c  \right] \notag \\
& \leq \frac{-\zeta_3}{n} \cdot \Ex\left[\widetilde{\Lambda}_3(i) \right] + e^{-n^{1/3}}\,. \notag
\end{align}
\end{small}
\end{proof}

The next proposition completes the proof of Theorem~\ref{thm:main}. We use the following lemma.

\begin{lemma}[Theorem A.1.15 in \cite{AlonSpencer_book}] \label{lem:alonspencer}
	Let $J$ have Poisson distribution with mean $\mu$. For $\epsilon > 0$
	\begin{align}
	\mathbb{P}\left[J \leq \mu(1-\epsilon)\right] & \leq e^{-\epsilon^2 \mu/2}, \notag \\
	\mathbb{P}\left[J \geq \mu(1+\epsilon)\right] & \leq \Bigl[ e^{\epsilon} (1 + \epsilon)^{-(1+\epsilon)}\Bigr]^{\mu}. \notag
	\end{align}
\end{lemma}

\begin{proposition}
Given $\theta > 0$, let $T_4 = T_4(\theta) = \frac{1 + 4\theta}{\zeta_{3}} \ln{n}$. Then with probability at least $1 - n^{-\theta/2}$, by time $T_4$ the system has reached consensus and the total communication until consensus is $O(\frac{n \log{n}}{s})$.
\end{proposition}
\begin{proof}
By Proposition~\ref{prop:exp_lambda_3_tilde}, for all $i \in \{nT_2, \ldots, 4n \ln{n}\}$, we have
\begin{small}
\begin{align}
\Ex\left[ \widetilde{\Lambda}_{3}(i+1) \right] & \leq  \left(1 - \frac{\zeta_3}{n} \right)\Ex\left[ \widetilde{\Lambda}_{3}(i) \right] + e^{-n^{1/3}}\,.
\end{align}
\end{small}
By definition of $\widetilde{\Lambda}_3$, we have that $ \widetilde{\Lambda}_{3}(i) \leq n$ for all $i$.
By induction on $i \geq nT_2$, it follows that
\begin{small}
\begin{align} \label{eq:lambda_3_ub_t2}
\Ex\left[ \widetilde{\Lambda}_{3}(i) \right] & \leq  \left(1 - \frac{\zeta_3}{n}\right)^{i-nT_2} \cdot n + (i-nT_2) \cdot e^{-n^{1/3}}\,.
\end{align}
\end{small}
Let $T_3 = T_2 + \frac{1 + \theta}{\zeta_3} \ln{n}$ as in Section~\ref{sec:phase_3}.
For $i = nT_3$, inequality (\ref{eq:lambda_3_ub_t2}) gives
\begin{small}
	\begin{align}
	\Ex\left[ \widetilde{\Lambda}_{3}(i) \right] & \leq  n\left(1 - \frac{\zeta_3}{n}\right)^{n(T_3 - T_2)}  + n(T_3 - T_2) \cdot e^{-n^{1/3}} \notag \\
	& \leq n e^{-(1+\theta)\log{n}} + \frac{1 + \theta}{\zeta_3} n \ln{n} \cdot e^{-n^{1/3}} < 2n^{-\theta}, \; \mbox{for large } n\,.
	\end{align}
\end{small}
If after $i$ clock rings there is no consensus, then at least one of $\widetilde{\alpha}(i)$, $\widetilde{\delta}(i)$, $\widetilde{\beta}_1(i), \ldots, \widetilde{\beta}_{8s+1}(i)$ is positive, so $\widetilde{\Lambda}_3(i) \geq \min\{1, d, a_1, \ldots, a_{8s+1}\} = a_{8s+1}$. Using Markov's inequality we get
\begin{align} \label{eq:prob_no_consensus}
\mathbb{P}(\mbox{no consensus after } n T_3 \mbox{ rings})
\leq \mathbb{P}\left(\widetilde{\Lambda}_3(n T_3) \geq a_{8s+1} \right) \leq
a_{8s+1}^{-1} \Ex\left( \widetilde{\Lambda}_3(n T_3) \right) \leq s 2^{8s} \cdot 2 n^{-\theta}
\end{align}

Let $J(t)$ denote the number of clock rings by time $t$. Then $J(t)$ has Poisson distribution with parameter $nt$. By Lemma~\ref{lem:alonspencer}, where $\mu = nT_4$ and $\epsilon = 1 - T_3/T_4$, we have
\begin{align} \label{eq:prob_few_clock_rings}
\mathbb{P}(J(T_4) \leq n T_3) & \leq \exp(- \left( 1 - \frac{T_3}{T_4}\right)^2 \cdot nT_4) \leq \exp\left(-\left(\frac{2\theta}{1 + 4 \theta}\right)^2 \cdot n\left(\frac{1+4\theta}{\zeta_3}\right) \ln{n}\right) \notag \\
& \leq \exp\left(- \frac{4 \theta^2}{\zeta_3(1+4\theta)} \cdot n \ln{n}\right) \,.
\end{align}

\smallskip

Combining (\ref{eq:prob_no_consensus}) and (\ref{eq:prob_few_clock_rings}) implies that
\begin{align} \label{eq:prob_no_consensus_ub_by_T4}
\mathbb{P}(\mbox{no consensus by time } T_4) & \leq \mathbb{P}(\mbox{no consensus after } n T_3 \mbox{ rings}) + \mathbb{P}\left(J(T_4) \leq n T_3\right)  \notag \\
& \leq 2s 2^{8s} n^{-\theta}  + \exp(-n\theta^2)
\end{align}

Let $L_i = \textbf{1}_{\{\mbox{i-th clock ring is a leader}\}}$ and $\Upsilon_i = \textbf{1}_{\{\mbox{i-th clock ring is uninformed}\}}$. Let $N_{C}(j)$ denote the number of communications by time $j$. Observe that the sequence $$M_j = \sum_{i=1}^{m} \Bigl(L_i - \Ex[L_i \mid \mathcal{F}_{i-1}] +  \Upsilon_i -\Ex[\Upsilon_i \mid \mathcal{F}_{i-1}]\Bigr) = \sum_{i=1}^{m} \left(L_i - \frac{1}{s} +  \Upsilon_i -\frac{\widetilde{u}(i)}{n}\right)$$ is a martingale, as the increments have mean zero given the past. Moreover, the increments are in $[-1, 1]$. Applying Azuma-Hoeffding gives
$\mathbb{P}\left(M_j \geq \frac{j}{2s}\right) \leq e^{-\frac{j}{8s^2}}$. Then we have
\begin{align} \label{eq:comm_by_nT_3_ub}
\mathbb{P}\left( N_C(nT_3) \geq \frac{2nT_3}{s} \right) & \leq \mathbb{P}\left( \sum_{i=1}^{n T_3} (L_i + \Upsilon_i) \geq \frac{2n T_3}{s} \right) \notag \\
& \leq \mathbb{P}\left(\sum_{i=1}^{n T_3} \left(L_i - \frac{1}{s} + \Upsilon_i - \frac{\widetilde{u}(i)}{n}\right) \geq \frac{ n T_3}{2s} \right) +  \mathbb{P}\left(\sum_{i=1}^{n T_3} \frac{\widetilde{u}(i)}{n} \geq \frac{n T_3}{2s} \right) \notag \\
& \leq e^{-\frac{n T_3}{8s^2}} + e^{-n^{1/3}} \leq e^{-{n}} + e^{-n^{1/3}} \,.
\end{align}
Combining inequalities (\ref{eq:prob_no_consensus_ub_by_T4}) and (\ref{eq:comm_by_nT_3_ub}) proves the proposition.
\end{proof}

\section{Concluding Remarks}

Recall the potential function $\Phi$ defined in Phase I. Next we discuss the extension of Theorem 1 to the setting where the initial advantage $\rho=\Phi(0)$ is allowed to tend to zero as $n \to \infty$.
\begin{corollary} Suppose $\rho=\Phi(0)>n^{-\sigma}$ for some $\sigma \in (0,1/2)$.
Then the protocol in Theorem 1 reaches consensus on the majority belief w.h.p in  time to $O(\log{n} +s|\log\rho|)$, and the total communication until consensus is $O(n/s \cdot \log(n)+ n |\log \rho|)$ w.h.p.
\end{corollary}
\begin{proof}
Let $\epsilon=\rho/16$ and $v_j= v_j  = \frac{2(3+ 2j)}{3+16s}$ for each $j \geq 1$. Recall the potential function from (\ref{psi_fun}):
\begin{align} 
\Psi(t) = \min\Bigl\{ \xi(t), \eta_1(t) + \epsilon \cdot v_1, \ldots, \eta_{8s}(t) + \epsilon \cdot v_{8s} \Bigr\}.
\end{align}
The proof of Proposition \ref{prop:phase1_progress} yields a lower bound for the right derivative
$\Psi_{+}'(t)  \ge \epsilon/(9s)=\rho/(144s)$, hence $t_1=288s$ satisfies
$\Psi(t_1) \ge 3\rho$, whence $\Phi(t_1)  \ge 5\rho/2>2\rho$. 

Iterating, we deduce that
$t_*=144s \log_2(1/\rho)$ satisfies $\Phi(t_*) \ge 1/2$. 
(Note that $\Psi$ changes its form in each round of the iteration, but $\Phi$ does not.) 
The assumption that $\rho>n^{-\sigma}$ for some $\sigma \in (0,1/2)$ is needed in order to apply Warnke's Theorem and conclude that the random system also satisfies $\tilde{\Phi}(t_1)>2\rho$ and reaches relative advantage at least $1/2$ by time $O(s\log(1/\rho))$.
\end{proof}


\bibliographystyle{natbib}
\bibliography{consensus_bib}

\end{document}